\documentclass[a4paper,11pt]{article}
\pdfoutput=1 % if your are submitting a pdflatex (i.e. if you have
\usepackage{jheppub} % for details on the use of the package, please
\usepackage[bottom]{footmisc}
\usepackage{amssymb}
\usepackage{amsmath}
\usepackage{amsthm}
\usepackage[usenames,dvipsnames]{xcolor}
\usepackage{epsfig}
\usepackage{dcolumn}
\usepackage{tikz}
\usetikzlibrary{shapes.geometric, arrows}
\usepackage{upgreek}
\usepackage{setspace}
\usepackage{enumitem}
\usepackage{array,multirow,bigdelim,arydshln}
\usepackage{appendix}
\usepackage{xparse}
\usepackage{stmaryrd}
\usepackage[T1]{fontenc} % if needed
\usepackage{mathtools}
\usepackage{physics} % Physiscs package
\usepackage{adjustbox}
\usepackage{multirow}
\usepackage{graphicx} % For images
\usepackage{float} % To manipulate figures
\graphicspath{{./images/}}
\usepackage[nottoc]{tocbibind}
\usepackage{hyperref}
\usepackage[utf8]{inputenc}
\hypersetup{
	colorlinks,
	urlcolor=Maroon,
	linkcolor=Maroon,
	citecolor=Maroon
	}

\NewDocumentCommand{\binomial}{omm}
 {%
  \genfrac(){0pt}{}{#2}{#3}%
  \IfValueT{#1}{_{\!#1}}%
 }
\NewDocumentCommand{\eulerian}{omm}
 {%
  \genfrac<>{0pt}{}{#2}{#3}%
  \IfValueT{#1}{_{\!#1}}%
 }

\def \s {\sigma}

\usepackage{latexsym}
\usepackage{tikz}

\theoremstyle{plain}

\newtheorem{thm}{Theorem}[section]
\newtheorem{prop}[thm]{Proposition}

\newtheorem{defn}[thm]{Definition}
\theoremstyle{definition}
\newtheorem{example}[thm]{Example}
\newtheorem{conj}[thm]{Conjecture}

\newcommand{\sfs}{\mathsf{s}}

\usepackage[percent]{overpic}
\usepackage{multirow} 
\usepackage{slashed}
\def\yz#1\yz {{\color{blue} [[YZ: #1]] }}
\def\bu#1\bu {{\color{red} [[BU: #1]] }}

\title{Connecting Scalar Amplitudes using The Positive Tropical Grassmannian}

\author[a]{Freddy Cachazo}\emailAdd{fcachazo@pitp.ca}
\author[a,b]{and Bruno Gim\'enez Umbert}\emailAdd{bgimenezumbert@pitp.ca}

\affiliation[a]{Perimeter Institute for Theoretical Physics, Waterloo, ON N2L 2Y5, Canada}
\affiliation[b]{Department of Physics and Astronomy, Western University, London, ON N6A 5B7, Canada}

\abstract{The biadjoint scalar partial amplitude, $m_n(\mathbb{I},\mathbb{I})$, can be expressed as a single integral over the positive tropical Grassmannian thus producing a {\it Global Schwinger Parameterization}. The first result in this work is an extension to all partial amplitudes $m_n(\alpha,\beta)$ using a limiting procedure on kinematic invariants that produces indicator functions in the integrand. The same limiting procedure leads to an integral representation of $\phi^4$ amplitudes where indicator functions turn into Dirac delta functions. Their support decomposes into $\textrm{C}_{n/2-1}$ regions, with $\textrm{C}_q$ the $q^{\rm th}$-Catalan number. The contribution from each region is identified with a $m_{n/2+1}(\alpha,\mathbb{I})$ amplitude. We provide a combinatorial description of the regions in terms of non-crossing chord diagrams and propose a general formula for $\phi^4$ amplitudes using the Lagrange inversion construction. We start the exploration of $\phi^p$ theories, finding that their regions are encoded in non-crossing $(p-2)$-chord diagrams. The structure of the expansion of $\phi^p$ amplitudes in terms of $\phi^3$ amplitudes is the same as that of Green functions in terms of connected Green functions in the planar limit of $\Phi^{p-1}$ matrix models. We also discuss possible connections to recent constructions based on Stokes polytopes and  accordiohedra.}

\begin{document}
\maketitle
\addtocontents{toc}{\protect\setcounter{tocdepth}{1}}
\def \tr {\nonumber\\}
\def \la  {\langle}
\def \ra {\rangle}
\def\hset{\texttt{h}}
\def\gset{\texttt{g}}
\def\sset{\texttt{s}}
\def \be {\begin{equation}}
\def \ee {\end{equation}}
\def \ba {\begin{eqnarray}}
\def \ea {\end{eqnarray}}
\def \k {\kappa}
\def \h {\hbar}
\def \r {\rho}
\def \l {\lambda}
\def \be {\begin{equation}}
\def \en {\end{equation}}
\def \bes {\begin{eqnarray}}
\def \ens {\end{eqnarray}}
\def \red {\color{Maroon}}
\def \pt {{\rm PT}}
\def \s {\sigma} % Is there no problem with line 117?
\def \ls {{\rm LS}}
\def \ma {\Upsilon}
\def \s {\textsf{s}}
\def \t {\textsf{t}}
\def \R {\textsf{R}}
\def \W {\textsf{W}}
\def \U {\textsf{U}}
\def \e {\textsf{e}}

\numberwithin{equation}{section}

\section{Introduction} \label{sec1}

The positive tropical Grassmannian $\textrm{Trop}^+G(2,n)$ \cite{SWTrop} is the space of all planar metric trees with $n$ leaves and vertices of any degree. Hence, in some sense it governs the singularity structure of any planar tree-level scattering amplitude of massless scalar fields with arbitrary polynomial interactions. 

For example, tree-level partial amplitudes of massless scalars in the biadjoint representation of $U(N)\times U(\tilde N)$ with only cubic interactions, $m_n(\alpha ,\beta)$, are computed by summing over all $n$-particle Feynman diagrams which are planar with respect to both, $\alpha$ and $\beta$, orderings. The partial amplitudes with the largest number of diagrams are the ones where both orderings coincide, e.g., $m_n(\mathbb{I},\mathbb{I})$. In a recent work by Cachazo and Early \cite{Cachazo:2020wgu}, a formula for $m_n(\mathbb{I},\mathbb{I})$ as a single integral over ${\rm Trop}^+G(2,n)$ was presented,
\be
m_n(\mathbb{I},\mathbb{I}) = \int_{\mathbb{R}^{n-3}}\!\! d^{n-3}x\, {\rm exp}(-F_n(x))\, .
\ee
Here the function $F_n(x)=  \sum_{a,b}t_{[a,b]}f_{[a,b]}(x)$ is a piece-wise linear function defined on $\textrm{Trop}^+G(2,n)$ via the ``tropical cross-ratios'' $f_{[a,b]}(x)$, and $t_{[a,b]}:=(p_a+p_{a+1}+\cdots +p_b)^2$ are standard planar kinematic invariants. This construction is reviewed in detail in section \ref{sec2}.

When restricted to a single cone of ${\rm Trop}^+G(2,n)$, the corresponding integral becomes nothing but the Schwinger parameterization of a single Feynman diagram. In this sense, ${\rm Trop}^+G(2,n)$ provides a global Schwinger parameterization of the amplitude. The regions in $\textrm{Trop}^+G(2,n)$ corresponding to trees with one or more higher-degree vertices are of measure zero and therefore do not contribute to the integral.

In this work we continue the study of global Schwinger parameterizations by extending the construction to all partial amplitudes $m_n(\alpha ,\beta)$ and to amplitudes in scalar theories with $\phi^p$-interactions and $p>3$. In both cases, our construction starts with the global Schwinger formulation of $m_n(\mathbb{I},\mathbb{I})$ and proceeds with a limiting procedure on the planar kinematic invariants to obtain  amplitudes in other theories. 

In the case of $m_n(\alpha ,\mathbb{I})$ partial amplitudes, the limiting procedure on kinematic invariants produces indicator functions in the integrand. These indicator functions describe the regions of the original $\textrm{Trop}^+G(2,n)$ that intersect with $\textrm{Trop}^+G_{\alpha}(2,n)$, defined with the ordering $\alpha$. Our first result is the following, 
\be
m_n(\alpha ,\mathbb{I}) = \int_{\mathbb{R}^{n-3}}\!\! d^{n-3}x\, {\rm exp}\left(-G_\alpha (x)\right) {\bf 1}_{S(H_\alpha )}(x)\,,
\ee
where $G_{\alpha}(x)$ is a piece-wise linear function that depends on the ordering $\alpha$, and ${\bf 1}_{S(H_\alpha )}(x)$ is an indicator function. We provide a derivation and examples of this formula in section \ref{sC}.

The next set of theories we study are a straightforward generalization of $m_n(\mathbb{I},\mathbb{I})$ in which one sums over only planar Feynman diagrams with $\phi^p$-interaction vertices. We denote such amplitudes by $A^{\phi^p}_n$. In this work we only consider planarity with respect to the cannonical ordering $\mathbb{I}$ and therefore there is no need to specify it. Of course, $A^{\phi^p}_n$ is defined to be zero if no Feynman diagram exists for the particular number of external particles.  

In order to obtain the global Schwinger formulation of $A^{\phi^p}_n$, the limiting procedure on kinematic invariants produces distributions in the integrand. In particular, the Dirac delta functions localize the integral over $\mathbb{R}^{n-3}$ to regions of measure zero\footnote{Each corresponds to a polyhedral cone in $\textrm{Trop}^+G(2,n)$.}, those where trees with higher-degree vertices live. For example, for $\phi^4$ amplitudes we find
\be
A^{\phi^4}_n =  \int_{\mathbb{R}^{n-3}}\!\! d^{n-3}x\, {\rm exp}\left(-\sum_{a<b}^{\rm even}t_{[a,b]}f_{[a,b]}(x)\right)Q(x)\,,
\ee
where $Q(x)$ is the distribution obtained from the limiting procedure and the sum is over $b-a\equiv 0\!\mod 2$. This procedure, together with the derivation of the global Schwinger formula for $A^{\phi^4}_n$, is developed in section \ref{sD}. In section \ref{examples4} we provide several examples.

Our global Schwinger formula for $\phi^4$ amplitudes reveals surprising connections to cubic amplitudes: we find that each of the $\textrm{C}_{n/2-1}$ regions that define the support of the distributions in the integrand is in bijection with a $m_{n/2+1}(\alpha ,\mathbb{I})$ amplitude. We study this feature in section \ref{combi}, where we propose a combinatorial procedure to obtain such regions from non-crossing chord diagrams. This implies that $A_n^{\phi^4}$ can also be expressed as a sum over regions. These results motivate a formula for the general {\it schematic} structure of $A_n^{\phi^4}$ in terms of cubic amplitudes based on the Lagrange inversion procedure
\be
A^{\phi^4}_n = \left(\frac{2}{n\, h_0^{n/2-1}}\right)\frac{1}{2\pi i}\oint_{|z|=\epsilon}dz\left(\frac{h(z)}{z}\right)^{n/2}~~{\rm with}~~h(x)=\sum_{i=0}^{\infty}m_{i+2}(\mathbb{I},\mathbb{I})\,x^i\,.
\ee
Here $h_0=m_2(\mathbb{I},\mathbb{I}):=P^2$ and $m_3(\mathbb{I},\mathbb{I}):=1$. 

The fact that $A^{\phi^4}_n$ is computed as a sum over regions is very reminiscent of the recent constructions based on Stokes polytopes \cite{yuliy,Banerjee:2018tun,Kalyanapuram:2019nnf,Aneesh:2019cvt,Salvatori:2019phs,Srivastava:2020dly}, which were motivated by the connection between $\phi^3$ amplitudes and the associahedron \cite{Mizera:2017cqs,Arkani-Hamed:2017mur}. It is known that some Stokes polytopes are associahedra and therefore their contribution could coincide with that of some of the regions we find. However, we find that only associahedra or intersections of associahedra \cite{Cachazo:2021llu} appear in our construction.  

In section \ref{phip} we start the exploration of $\phi^p$ amplitudes in general. We propose an analogous limiting procedure and obtain the corresponding global Schwinger formula
\be
A_n^{\phi^p}=\int_{\mathbb{R}^{n-3}}d^{n-3}x\,\textrm{exp}\left(-\sum_{a<b}^{{\cal K}_p}t_{[a,b]}f_{[a,b]}(x)\right)Q(x)\,.
\ee
where ${\cal K}_p$ indicates that the sum is over ordered pairs $(a,b)$ such that $b-a\equiv 0\!\! \mod p-2$.

We also propose a diagrammatic construction to find the regions that compute $A_n^{\phi^p}$ as non-crossing $(p-2)$-chord diagrams (these are counted by the Fuss-Catalan numbers\footnote{Here $\textrm{FC}_m(q,r)$ is the Fuss-Catalan number given by $$\textrm{FC}_m(q,r)\equiv \frac{r}{mq+r} {mq+r \choose m}\,.$$} $\textrm{FC}_{(n-2)/(p-2)}(p-2,1)$). The sum over all contributions leads to the expected number of trees in $A_n^{\phi^p}$, which is also given by Fuss-Catalan numbers, $\textrm{FC}_{(n-2)/(p-2)}(p-1,1)$. Moreover, we point out a connection to $m_{(n+2(p-3))/(p-2)}(\alpha,\mathbb{I})$ amplitudes and provide some examples. We also propose a formula giving the general structure of $A_n^{\phi^p}$ in terms of cubic amplitudes. 

Of course, $A_n^{\phi^p}$ amplitudes have also been recently studied and found to be related to a class of polytopes known as accordiohedra \cite{2017arXiv170309953M,Raman:2019utu,Aneesh:2019ddi,Kojima:2020tox,Kalyanapuram:2020vil,John:2020jww,Kalyanapuram:2020axt,Jagadale:2021iab} (see also \cite{Baadsgaard:2015ifa,Baadsgaard:2016fel} for related work). Some accordiohedra are associahedra and therefore as in the case of $\phi^4$ we suspect that contributions from such accordiohedra could coincide with that of some of our regions. 

In section \ref{discs} we conclude with discussions on possible future research directions including connections between our schematic formulas for $A_n^{\phi^p}$ and those that express general Green functions in terms of connected Green functions in planar theories, a way to connect to accordiohedra constructions, and possible extensions to CEGM generalized amplitudes. We end this work with two appendices. In appendix \ref{lagrange} we give a review of the Lagrange inversion formula and provide details on the derivations of the schematic structures of $A_n^{\phi^p}$ amplitudes. In appendix \ref{appB} we provide an example which uses the global Schwinger formula to compute the contribution to regions without actually carrying out the integral by identifying the corresponding $\phi^3$ amplitude.

\section{Global Schwinger Formula for $m_n(\mathbb{I},\mathbb{I})$}\label{sec2}

In this section we review the construction of the global Schwinger formula for $m_n(\mathbb{I},\mathbb{I})$ introduced and proved in \cite{Cachazo:2020wgu}. 

Consider a single metric tree ${\cal T}$ with $n$ leaves and all internal vertices of degree three. We follow the mathematical convention and call these binary trees\footnote{The name stems from the fact that when a leaf is selected as a root, then walking up along the tree implies that at each internal vertex there are exactly two possible edges to choose from in order to continue the walk.}. 

Label the leaves of ${\cal T}$ so that it is planar with respect to the ordering $\mathbb{I}:= (1,2,\ldots ,n)$. Its contribution to an amplitude can be constructed as follows. First, define the function
\begin{align}\label{oneF}
    {\cal F}({\cal T}):=-\sum_{1\leq a,b \leq n}d_{ab}s_{ab}\,,
\end{align}
where $d_{ab}$ represents the matrix of distances, i.e., the distance from leaf $a$ to leaf $b$. Mandelstam invariants $s_{ab}:=(k_a+k_b)^2$ satisfy 
\begin{align}\label{con}
   s_{ab}=s_{ba},\quad s_{aa}=0, \quad {\rm and}\quad  \sum_{b=1}^ns_{ab}=0 \quad  \forall a.
\end{align}
Let $e_a$ be the length of the edge containing the $a^{\rm th}$ leaf and write $d_{ab}=e_a+e_b+d^{\textrm{int}}_{ab}$ where $d^{\textrm{int}}_{ab}$ is the length of the internal edges along the unique path connecting $a$ and $b$. Due to momentum conservation \eqref{con}, $e_a$ drops out from the function ${\cal F}({\cal T})$ and it can be written as 
\begin{align}
    {\cal F}({\cal T})=\sum_{i=1}^{n-3}f_{I_i} t_{I_i} \,,
\end{align}
where $f_I$ denotes the length of an internal edge that partitions the set leaves of ${\cal T}$ as $I\cup I^{\rm c} =[n]$. The kinematic invariant multiplying $f_I$ is defined as the square of the momentum flowing through the edge under consideration, i.e.,
\be 
t_I := \left(\sum_{a\in I}k_a\right)^2 =\sum_{\{a,b\}\subset I}s_{ab}\,.
\ee
The conditions \eqref{con} guarantee that $t_I = t_{I^{\rm c}}$.

Finally, the contribution to the amplitude is
\begin{align}\label{SR1}
    {\cal R}({\cal T})=\int_{O^+}\!\!\!\! d^{n-3}f \, {\rm exp}(-{\cal F}({\cal T})) = \prod_{i=1}^{n-3}\int_0^{\infty}\!\!\! df_{I_i}\,\textrm{exp}(-f_{I_i}t_{I_i}) = \prod_{i=1}^{n-3}\frac{1}{t_{I_i}} \,.
\end{align} 
where $O^+:=(\mathbb{R}^+)^{n-3}$ is the positive orthant in $\mathbb{R}^{n-3}$. Of course, the integral formulas are only defined for $t_{I_i}>0$ but once the answer is in the rational function form, it is valid for any values $t_{I_i}\neq 0$.

Note that the second integral in \eqref{SR1} is the standard Schwinger formula and the edge lengths $f_{I}$ are the Schwinger parameters.

Denoting the set of all binary trees which are planar with respect to an ordering of the leaves $\alpha :=(\alpha_1,\alpha_2,\ldots ,\alpha_n)$ by ${\tt Pl}_n(\alpha)$, the amplitude is computed as
\be\label{amp1}
m_n(\mathbb{I},\mathbb{I}) = \!\! \sum_{{\cal T}\in\, {\tt Pl}_n(\mathbb{I})} {\cal R}({\cal T})\,.
\ee 

In \cite{SSTrop}, Speyer and Sturmfels introduced the tropical Grassmannian ${\rm Trop}\, G(2,n)$ and showed that it agrees with the moduli space of phylogenetic trees studied by Billera, Holmes and Vogtmann (BHV) \cite{BHV}. 

Motivated by the work of Postnikov \cite{Alex} on totally positive Grassmannians, Speyer and Williams introduced positive tropical Grassmannians \cite{SWTrop}. In particular, ${\rm Trop}^+\, G(2,n)$ parameterizes the space of planar trees on $n$ leaves. This means that ${\rm Trop}^+\, G(2,n)$ must provide a global definition of Schwinger parameters which unifies all the individual Schwinger representations into a single integral.

In order to present the formula, one starts with $G^+(2,n)$ and then tropicalizes the Pl\"ucker coordinates. Such a positive parameterization of $G^+(2,n)$ is given by\footnote{Here we suppress the torus coordinates which under tropicalization map to the $e_a$'s which drop out.} 
\begin{align}\label{preT}
\begin{pmatrix}
1 & 0 & -1 & -(1+\tilde{x}_1) & -(1+\tilde{x}_1+\tilde{x}_2) & \cdots & -(1+\tilde{x}_1+\cdots+\tilde{x}_{n-3})\\
0 & 1 & 1 & 1 & 1 & \cdots & 1
\end{pmatrix}
\end{align}
where $\tilde{x}_a\in\mathbb{R}^+$. Note that any minor $\Delta_{ab}$ with $a<b$ is positive.

The tropicalization of a minor $\Delta_{ab}$ of \eqref{preT} proceeds by replacing addition, ${\tilde x_i}+{\tilde x_j}$, with the min-function, ${\rm min}(x_i,x_j)$, and multiplication ${\tilde x_i}{\tilde x_j}$ with addition $x_i+x_j$. Note that we drop the tilde to differentiate the two sets of variables. This is important since while $\tilde{x}_a\in\mathbb{R}^+$, the tropical variables are unconstrained, i.e., $x_a\in\mathbb{R}$.

The connection to the space of planar trees is simply the identification of (minus) the distance matrix $d_{ab}$ with the tropical Pl\"ucker coordinates. It is not difficult to evaluate $\Delta^{\rm Trop}_{ab}(x)$ and find
\be\label{tropMin} 
-d_{ab} \leftrightarrow \Delta^{\rm Trop}_{ab}(x) = \left\{ \begin{array}{ccc}
 {\rm min}(x_{a-2},x_{a-1},\ldots, x_{b-3}) &\,\,\,\,  & 2\leq a\leq b-1,\, 4\leq b \leq n \\
 0 & & {\rm otherwise}
\end{array}\right. , 
\ee 
where $x_0:=0$ and whenever there is a single argument ${\rm min}(x) := x$.

It might seem strange that some $d_{ab}$ are sent to zero, however, recall that momentum conservation makes the physics independent of the lengths $e_a$ which can then be used to set to zero some of the entries $d_{ab}$. The choice corresponds to a choice of frame in \eqref{preT}.

Using this in \eqref{oneF} one defines the ``tropical potential function''
\begin{align}\label{tropP}
    F_n(x):=\sum_{1\leq a < b \leq n}\!\!\! s_{ab}\, \Delta^{\rm Trop}_{ab}(x) =\sum_{b=4}^n\sum_{a=2}^{b-1}\! s_{ab}\,{\rm min}(x_{a-2},x_{a-1},\ldots, x_{b-3})\,.
\end{align}

The scattering amplitude \eqref{amp1} now has a single integral representation \cite{Cachazo:2020wgu}
\be\label{amp2}
m_n(\mathbb{I},\mathbb{I}) = \int_{\mathbb{R}^{n-3}}\!\! d^{n-3}x\, {\rm exp}(-F_n(x))\,.
\ee 
Note that the integral is over all $\mathbb{R}^{n-3}$. 

Of course, the integral in \eqref{amp2} might not exist for some values of kinematic invariants. Let us discuss the regions of convergence. The best approach is to write the tropical potential function $F_n(x)$ in terms of planar kinematic invariants. 

This is easily done by introducing the notation $t_{[a,b]}$ to denote $t_I$ with $I=\{ a,a+1,\ldots ,b-1,b \}$, a set of consecutive labels, and using
\be\label{fourTerms} 
s_{ab} = t_{[a,b]}-t_{[a+1,b]}-t_{[a,b-1]}+t_{[a+1,b-1]}\,.
\ee 
Here one defines $t_{[c,d]}=0$ whenever $c\geq d$. For example, $s_{14}=t_{[1,4]}-t_{[2,4]}-t_{[1,3]}+t_{[2,3]}$ while $s_{23}=t_{[2,3]}$. 

Using \eqref{fourTerms} in $F_n(x)$ and arranging by planar kinematic invariants one finds
\be\label{PlanarF} 
F_n(x) = \sum_{a<b}t_{[a,b]}\left( \Delta^{\rm Trop}_{a,b}(x)-\Delta^{\rm Trop}_{a,b+1}(x)-\Delta^{\rm Trop}_{a-1,b}(x)+\Delta^{\rm Trop}_{a-1,b+1}(x) \right)\,.
\ee 

\begin{figure}[h!]
	\centering
	\includegraphics[width=0.30\linewidth]{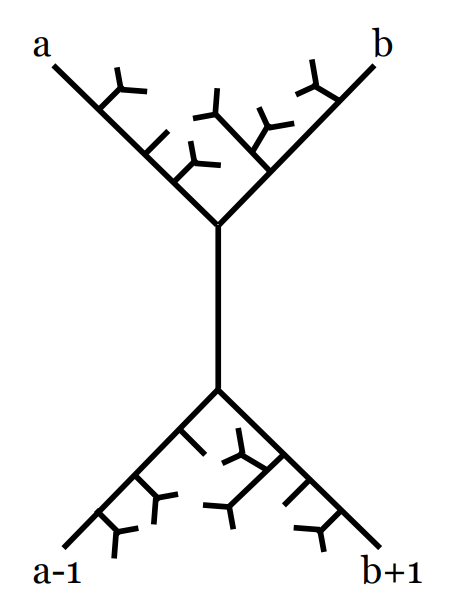}
	\caption{In a generic tree, the combination of distances $-d_{a,b}+d_{a,b+1}+d_{a-1,b}-d_{a-1,b+1}$ equals twice the length of the edge which removal would split the diagram into two parts, one containing $a$ and $b$ and the other $a-1$ and $b+1$.}
	\label{fig:abTree}
\end{figure}

The quantity in brackets has a very beautiful interpretation when thought of as $-d_{a,b}+d_{a,b+1}+d_{a-1,b}-d_{a-1,b+1}$ for a single planar Feynman diagram. This is nothing but twice the length of the edge partitioning the labels as $\{a,a+1,\ldots ,b-1,b\}\cup \{ b+1,b+2,\ldots, a-2,a-1\}$, i.e., what used to be $f_{[a,b]}$, see figure \ref{fig:abTree}. Of course, for this to be the case, it better be that it is always non-negative. This is easily seen to be the case by noticing the following general property for any three real numbers $A,B,C$,
\be  
A-{\rm min}(A,B)-{\rm min}(A,C)+{\rm min}(A,B,C) \geq 0\,.
\ee 
The proof is left as an exercise to the reader\footnote{Hint: $A-{\rm min}(A,B)=-{\rm min}(A-A,B-A)=-{\rm min}(0,B-A) \geq 0$. Repeated use leads to  $-{\rm min}(0,B-A)-{\rm min}(0,(B-A)-{\rm min}(0,C-A))$.}.

It is important to mention that the condition 
\be 
\Delta^{\rm Trop}_{a,b}(x)+\Delta^{\rm Trop}_{a-1,b+1}(x) \geq \Delta^{\rm Trop}_{a,b+1}(x)+\Delta^{\rm Trop}_{a-1,b}(x)
\ee 
is part of what is referred to as a positive tropical Pl\"ucker relation and it must be satisfied in order to be in ${\rm Trop}^+G(2,n)$. 

Finally, the condition for the integral formula \eqref{amp2} to exist is simply that all planar kinematic invariants be positive. In the rest of this work, only the formula with planar invariants will be used. This why it is convenient to introduce special notation for the combination of tropical minors in \eqref{PlanarF}, 
\be\label{defFint} 
f_{[a,b]}(x):=\Delta^{\rm Trop}_{a,b}(x)-\Delta^{\rm Trop}_{a,b+1}(x)-\Delta^{\rm Trop}_{a-1,b}(x)+\Delta^{\rm Trop}_{a-1,b+1}(x)\,, 
\ee 
so that $F_n(x)=  \sum_{a,b}t_{[a,b]}f_{[a,b]}(x)$ and we arrive at the final form of the global Schwinger formula \cite{Cachazo:2020wgu}
\be\label{amp3}  
m_n(\mathbb{I},\mathbb{I}) = \int_{\mathbb{R}^{n-3}}\!\! d^{n-3}x\, {\rm exp}\left(-\sum_{a< b}t_{[a,b]}f_{[a,b]}(x)\right)\,.
\ee

\section{From $m_n(\mathbb{I},\mathbb{I})$ to $m_n(\alpha , \beta)$}\label{sC}

In order to extend the global Schwinger formula \eqref{amp3} for $m_n(\mathbb{I},\mathbb{I})$ to all other partial amplitudes, $m_n(\alpha , \beta)$, let us first review their definition. Recall that ${\tt Pl}_n(\alpha)$ denotes the set of all binary trees which are planar with respect to the ordering of the leaves defined by $\alpha =(\alpha_1,\alpha_2,\ldots ,\alpha_n)$. Given a second ordering $\beta =(\beta_1,\beta_2,\ldots ,\beta_n)$, one can determine the set of trees which are planar with respect to both orderings by simply finding the intersection ${\tt Pl}_n(\alpha)\cap {\tt Pl}_n(\beta)$, and therefore,
\be\label{ampAll}
m_n(\alpha , \beta) = \!\! \sum_{{\cal T}\in\, {\tt Pl}_n(\alpha)\cap \, {\tt Pl}_n(\beta)} {\cal R}({\cal T})\,.
\ee 
Depending on conventions, there might be an overall sign which depends on the two orderings chosen. Since our main concern is the kinematic dependence of the amplitude, we refer the reader to \cite{Cachazo:2013iea} for details on the definition of the sign\footnote{For example, one could decide to define partial amplitudes so that the sign is included in the traces of the flavour groups. While convenient when individual partial amplitudes are considered, this makes properties such as the $U(1)$-decoupling identity, which involves several partial amplitudes, cumbersome.}. 

Without loss of generality we assume that $\beta =\mathbb{I} =(1,2,\ldots ,n)$. Now, recall that $t_I$ is planar with respect to an ordering if the set $I$ coincides with the set of labels of an interval in the ordering. For example, if $\alpha = (1,2,5,4,3,6)$ then $I=\{ 1,2,5\}$ is planar with respect to $\alpha$ but not with respect to $\mathbb{I}$ while $I=\{ 3,4,5\}$ is planar with respect to both orderings. 

Let ${\tt PK}(\alpha)$ denote the set of all planar kinematic invariants with respect to $\alpha$. 

\begin{prop}\label{limit1}
Consider the set of planar kinematic invariants ${\tt PK}(\mathbb{I})$, set  $t_I=1/\epsilon$ whenever $t_I \notin {\tt PK}(\alpha)$, and evaluate $m_n(\mathbb{I},\mathbb{I})$ on it to get a function of $\epsilon$ and kinematic invariants in  ${\tt PK}(\alpha)\cap {\tt PK}(\mathbb{I})$, $m_n^{(\epsilon )}(\mathbb{I},\mathbb{I})$. Then
\be\label{amp9} 
m_n(\alpha ,\mathbb{I}) = \lim_{\epsilon\to 0} m_n^{(\epsilon )}(\mathbb{I},\mathbb{I})\,.
\ee 
\end{prop}

\begin{proof}
Since all trees in $m_n^{(\epsilon )}(\mathbb{I},\mathbb{I})$ that do not contribute to $m_n(\alpha ,\mathbb{I})$ contain at least one kinematic invariant that has been set to $1/\epsilon$, in the limit their contribution to the amplitude vanishes. Since invariants in ${\tt PK}(\alpha)\cap {\tt PK}(\mathbb{I})$ are $\epsilon$-independent, so are the corresponding Feynman diagram contributions ${\cal R}({\cal T})$ to \eqref{ampAll}.
\end{proof}

The construction of the global Schwinger formula for $m_n(\alpha ,\mathbb{I})$ proceeds in exactly the same way. Let us define the $\epsilon$-dependent tropical potential function, $F_n(x,\epsilon)$, by starting with $F_n(x)$ and restricting to the kinematic space of Proposition \ref{limit1}. More explicitly,
\be 
F_n(x) := \sum_{I \in \, {\tt PK}(\mathbb{I})} t_I\, f_I(x) \to F_n(x,\epsilon ) := \!\!\! \sum_{I \in {\tt PK}(\alpha)\cap\, {\tt PK}(\mathbb{I})}\!\!\! t_I \, f_I(x) + \frac{1}{\epsilon }\left( \sum_{I \notin {\tt PK}(\alpha)\cap \, {\tt PK}(\mathbb{I})} \!\!\! f_I(x)\right)\,.
\ee 
Let us define the finite and divergent parts to be $F_n(x,\epsilon ) = G_\alpha (x)+ \frac{1}{\epsilon} H_\alpha (x)$, i.e.,
\be\label{aux1} 
G_\alpha (x):=\sum_{I \in {\tt PK}(\alpha)\cap\, {\tt PK}(\mathbb{I})}\!\!\! t_I \, f_I(x)\,,\quad H_\alpha (x):= \sum_{I \notin {\tt PK}(\alpha)\cap \, {\tt PK}(\mathbb{I})} \!\!\! f_I(x)\,.
\ee 
Note that we have chosen to add the subscript $\alpha$ to indicate that the form of the functions depends on the $\alpha$-ordering.

Using \eqref{amp9} and \eqref{amp2} one finds
\be 
m_n(\alpha ,\mathbb{I}) = \lim_{\epsilon\to 0^+} \int_{\mathbb{R}^{n-3}}\!\! d^{n-3}x\, {\rm exp}(-F_n(x,\epsilon))\,.
\ee
Defining the limit as a directional limit from above is necessary for the convergence of the integral since all planar kinematic invariants must be positive. Moreover, it also allows the limit to be taken inside the integral. Using \eqref{aux1} one finds
\be 
m_n(\alpha ,\mathbb{I}) = \int_{\mathbb{R}^{n-3}}\!\! d^{n-3}x\, {\rm exp}\left(-G_\alpha (x)\right) \lim_{\epsilon\to 0^+}{\rm exp}\left(-\frac{1}{\epsilon}H_\alpha (x)\right)\,.
\ee
The function resulting from computing the limit is nothing but an indicator function. In general, given two sets $S,U$ such that $S\subset U$, 
\be\label{indicatorDef} 
{\bf 1}_{S}: U\to \{0,1\},\quad {\bf 1}_{S}(x)=\left\{ \begin{array}{cc} 1 & {\rm if~} x \in S, \\ 0 & {\rm ~\, otherwise}. \end{array} \right.
\ee 
In the case at hand, we define the set $S(H_\alpha ) := \{ x\in \mathbb{R}^{n-3} : H_\alpha (x) = 0 \}\subset \mathbb{R}^{n-3}$. This leads to the final formula for the global Schwinger formula,
\be \label{ffS} 
m_n(\alpha ,\mathbb{I}) = \int_{\mathbb{R}^{n-3}}\!\! d^{n-3}x\, {\rm exp}\left(-G_\alpha (x)\right) {\bf 1}_{S(H_\alpha )}(x)\,.
\ee 

\begin{example}\label{ex1}

Consider $\alpha = (1324)$. In this case
\be  
G_\alpha (x) = -s_{23}\, {\rm min}(0,x_1)\,, \quad H_\alpha (x) = x_1-{\rm min}(0,x_1)\,.
\ee 
The set $S(H_\alpha )= \{ x_1 : x_1\leq 0\} = (-\infty , 0]$ and therefore,
\be 
m_4(1324,\mathbb{I}) = \int_{-\infty}^\infty\!\! d x_1  \, {\rm exp}\left(s_{23}\, {\rm min}(0,x_1)\right){\bf 1}_{(-\infty , 0]}(x_1) =  \int_{-\infty}^0\!\! d x_1 \, {\rm exp}\left(s_{23}\, x_1\right) = \frac{1}{s_{23}}\,.
\ee 
\end{example}

Let us present a more interesting example.

\begin{example}\label{ex2}

Consider $\alpha = (123654)$. In this case
\begin{align}
\begin{split}
    G_\alpha (x) = & s_{12} \left(x_1-\min
   \left(0,x_1\right)\right)-s_{23}\, \min \left(0,x_1\right)+s_{45} \left(x_2-\min
   \left(x_2,x_3\right)\right)+\\ & s_{56} \left(x_3-\min
   \left(x_2,x_3\right)\right)+t_{123} \left(\min \left(x_2,x_3\right)-x_1\right), \\
   H_\alpha (x) = & x_1-\min (x_1,x_2,x_3)\,.
\end{split}  
\end{align}
In the expression for $G_\alpha(x)$ we have already used that $x_1\leq \min (x_2,x_3)$ is the condition imposed by requiring $H_\alpha (x)=0$ in order to simplify the expression. The set $S(H_\alpha )= \{ x_1 : x_1\leq \min (x_2,x_3) \}$. In this case it is convenient to write the indicator function as a product of two Heaviside step functions $\theta (x_2-x_1)\theta(x_3-x_1)$ so that 
\be 
m_6(123654,\mathbb{I}) = \int_{\mathbb{R}^3}\!\! d^3x  \, {\rm exp}\left(-G_\alpha(x_1,x_2,x_3)\right)\theta (x_2-x_1)\theta(x_3-x_1)\,. 
\ee 
This integral is easily evaluated to give the expected result
\be 
m_6(123654,\mathbb{I}) = \left(\frac{1}{s_{12}}+\frac{1}{s_{23}}\right)\left(\frac{1}{s_{45}}+\frac{1}{s_{56}}\right)\frac{1}{t_{123}}\,. 
\ee 
\end{example}

\section{From $\phi^3$ Amplitudes to $\phi^4$ Amplitudes} \label{sD}

The positive tropical Grassmannian ${\rm Trop}^+\, G(2,n)$ is the space of all planar metric trees. In other words, trees with vertices of any degree $3\leq d\leq n$ are part of the space. In the previous section, amplitudes of theories where only Feynman diagrams corresponding to binary trees were discussed. At first it might be puzzling that a formula for $m_n(\mathbb{I},\mathbb{I})$ involves an integration over the entire ${\rm Trop}^+\, G(2,n)$. However, this is easily understood by noticing that the regions in ${\rm Trop}^+\, G(2,n)$ which correspond to trees with at least one vertex of degree $d>3$ are of measure zero and do not contribute to the integral. 

In this section we extend the idea used to obtain $m_n(\alpha,\mathbb{I})$ from $m_n(\mathbb{I},\mathbb{I})$ by a limiting procedure in order to obtain a global Schwinger formula for $A^{\phi^4}_n$. The main difference is that while the limiting procedure produced indicator functions leading to  $m_n(\alpha,\mathbb{I})$, here it produces Dirac delta functions that localize the integral to the regions of measure zero where $\phi^4$ planar trees are located. The process unearths a surprising connection to $m_{n/2+1}(\alpha ,\mathbb{I})$ amplitudes. 

\begin{prop}\label{qeri}
Consider the space of kinematic invariant of $n=2m$ massless particles with $t_{[a,b]}= 1/\epsilon$ whenever $b-a\equiv 1\!\! \mod 2$ and let $m_n^{(\epsilon)}(\mathbb{I},\mathbb{I})$ denote $m_n(\mathbb{I},\mathbb{I})$ evaluated on it. Then 
\be 
A^{\phi^4}_n = \lim_{\epsilon\to 0}\frac{1}{(2\epsilon)^{n/2-1}}m_n^{(\epsilon)}(\mathbb{I},\mathbb{I})\,.
\ee 
\end{prop}

\begin{proof}

Consider any Feynman diagram ${\cal T}^{(4)}$ contributing to $A^{\phi^4}_n$, that is, any completely ternary planar tree on $n$-leaves. Such a diagram has $n/2-1$ vertices of degree $4$. The strategy is to find out how many planar binary trees give rise to ${\cal T}^{(4)}$ by collapsing edges, i.e., taking their length to zero. This is easily done by realizing that for each degree-four vertex of ${\cal T}^{(4)}$ there are exactly two ways, compatible with planarity, of growing an edge to produce two degree-three vertices. This means that there are $2^{n/2-1}$ binary trees that give rise to ${\cal T}^{(4)}$. Of course, not all binary trees descend to a ternary diagram. Note that under the kinematics in the proposition, $\phi^3$ Feynman diagrams that collapse to $\phi^4$ diagrams have exactly $n/2-2$ $\epsilon$-independent propagators and $(n-3)-(n/2-2)=n/2-1$ which become $1/t = \epsilon$. Diagrams that do not produce a $\phi^4$ diagram have at least one extra propagator of the form $1/t=\epsilon$. Therefore, in the limit $\epsilon\to 0 $ the $\phi^4$ amplitude is recovered. 
\end{proof}

\subsection{Global Schwinger Formula for $A^{\phi^4}_n$ }

Following the same steps as in section \ref{sC} we start the derivation of the global Schwinger formulation of $A^{\phi^4}_n$ by using Proposition \ref{qeri} and the representation for $m_n(\mathbb{I},\mathbb{I})$ given in \eqref{amp3}, i.e.
\be\label{amp44}  
A^{\phi^4}_n = \lim_{\epsilon\to 0^+}\frac{1}{(2\epsilon)^{n/2-1}} \int_{\mathbb{R}^{n-3}}\!\! d^{n-3}x\, {\rm exp}\left(-\sum_{a< b}t_{[a,b]}(\epsilon)f_{[a,b]}(x)\right)\,.
\ee 
Rewrite
\be
F_n(x,\epsilon ) = \sum_{a< b}t_{[a,b]} f_{[a,b]}(x) = G(x)+\frac{1}{\epsilon}H(x)
\ee 
with
\be\label{Hregions}  
G(x) = \sum_{a<b}^{\rm even}t_{[a,b]}f_{[a,b]}(x)\,,\quad H(x) = \sum_{a<b}^{\rm odd}f_{[a,b]}(x)\,,
\ee 
where the sums are over ordered pairs $(a,b)$ such that 
$b-a\equiv 0\!\mod 2$ (even) or $b-a\equiv 1\!\mod 2$ (odd).

Unlike the cases considered in the previous section, the limit $\epsilon\to 0$ of  $m_n^{(\epsilon)}(\mathbb{I},\mathbb{I})$ does not lead to a finite answer and therefore commuting the limit and the integration in \eqref{amp44} must be carefully defined. 

We take the approach in which
\be\label{distrib}  
 Q(x):= \lim_{\epsilon\to 0^+}\frac{1}{(2\epsilon)^{n/2-1}} \exp \left( -\frac{1}{\epsilon}H(x)\right)
\ee 
is to be treated as a distribution. Since $H(x)\geq 0$, it is clear that $Q(x)$ only has support in regions where $H(x)=0$.

In the next section we show that solutions to $H(x)=0$ are regions of dimension $n/2-1$ in $\mathbb{R}^{n-3}$ which are classified by non-crossing chord diagrams. Here we show the explicit form of $H(x)$,
\begin{align}\label{genH}
H(x) = \sum_{a=0}^{n-3}x_a+2\sum_{a=0}^{n-4}\sum_{b=a+1}^{n-3}(-1)^{b-a}\,\textrm{min}(x_a,x_{a+1},...,x_b)\,.
\end{align} 

Thus, the distribution $Q(x)$ becomes a sum over distributions that localize the integral to the regions. This gives the first form of the global Schwinger formula for the $\phi^4$ theory,
\be\label{amp56}  
A^{\phi^4}_n =  \int_{\mathbb{R}^{n-3}}\!\! d^{n-3}x\, {\rm exp}\left(-\sum_{a<b}^{\rm even}t_{[a,b]}f_{[a,b]}(x)\right)Q(x)\,.
\ee 
In the next section we present some examples that motivate a second version of the formula as a sum over regions labelled by non-crossing chord diagrams.

\section{Computing $\phi^4$ Amplitudes Using the Global Schwinger Formula}\label{examples4}

In this section we illustrate the use of the global Schwinger formula \eqref{amp56} by considering several examples. 

\subsection{Four-Point Amplitude}

The four-particle kinematic space is only two dimensional, $s_{12},s_{23}$. Therefore $G(x)=0$ and
\be 
H(x) = x_1-2\min (0,x_1) = |x_1|.
\ee 
The distribution in the integral is 
\be 
Q(x) = \lim_{\epsilon\to 0^+}\frac{1}{2\epsilon} \exp \left( -\frac{1}{\epsilon}|x_1|\right) = \delta(x_1)\,.
\ee 
This implies that \eqref{amp56} becomes
\be 
A^{\phi^4}_4 =  \int_{\mathbb{R}}\!\! dx_1 \delta(x_1) = 1\,.
\ee 

\subsection{Six-Point Amplitude}

The six-particle kinematic space is nine dimensional, $s_{12},s_{23},\ldots ,s_{61}, t_{[1,3]},t_{[2,4]},t_{[3,5]}$. Evaluating $H(x)$ one finds
\begin{align}
    H(x) = & \, x_1+x_2+x_3-2\min (0,x_1)-2\min (x_1,x_2)-2\min (x_2,x_3)\nonumber \\ & +2\min (0,x_1,x_2)+2\min (x_1,x_2,x_3)-2\min (0,x_1,x_2,x_3)\,.
\end{align} 
Setting $H(x)$ to zero gives rise to two regions,
\be
R_1 = \{ x_1=0,\, x_2 = x_3 \}\,, \quad R_2 = \{ x_1=x_2 > 0, \, x_3 = 0 \}\, .\label{region6}
\ee 
The distribution $Q(x)$ then becomes 
\be 
Q(x) =Q_1(x)+Q_2(x)\,,\,\, {\rm with}\,\, Q_1(x):= \delta(x_1)\delta(x_2-x_3)\,, \,\, Q_2(x) := \theta(x_1)\delta(x_1-x_2)\delta(x_3)\,.
\ee 
Instead of computing \eqref{amp56} as a single object, let us split it by regions
\be\label{our1} 
A^{\phi^4: (1)}_6 =  \int_{\mathbb{R}^3}\!\! d^3 x\exp(-G(x)) Q_1(x) = \frac{1}{t_{123}}+\frac{1}{t_{234}}\,.
\ee 
\be\label{our2} 
A^{\phi^4: (2)}_6 =  \int_{\mathbb{R}^3}\!\! d^3 x\exp(-G(x)) Q_2(x) = \frac{1}{t_{345}}\,.
\ee 
Adding up the two contributions leads to the amplitude
\be  
A^{\phi^4}_4 = \frac{1}{t_{123}}+\frac{1}{t_{234}}+\frac{1}{t_{345}}\,.
\ee 

\subsection{Eight-Point Amplitude}\label{sec:Eight_Point_Amplitude}

The eight-particle kinematic space is twenty dimensional, $s_{12},s_{23},\ldots ,s_{81}$, $t_{123},t_{234},\ldots ,t_{812}$, and $t_{[1,4]},t_{[2,5]},t_{[3,6]},t_{[4,7]}$. Evaluating $H(x)$ using \eqref{genH} one finds five regions:
\begin{align}\label{eight4}
    R_1 = \,& \{ x_1=0,\, x_2 = x_3,\, x_4= x_5 \}\,, \nonumber \\  
    R_2 = \,& \{  x_1=0, \, x_2=x_5,\, x_3=x_4,\, x_2 < x_3 \}\,, \nonumber \\ 
    R_3 = \,& \{ x_5= 0,\, x_1=x_2,\, x_3=x_4,\, x_1 > 0,\, x_3 > 0 \}\,, \nonumber \\ 
    R_4 = \,& \{ x_3 =0,\, x_1=x_2,\, x_4=x_5,\, x_1>0 \} \,,\nonumber \\
    R_5 = \,& \{ x_5 = 0 ,\, x_1=x_4,\, x_2=x_3,\, x_2>x_1>0 \}\,.
\end{align} 
The distribution $Q(x)$ then becomes 
\be 
Q(x) =Q_1(x)+Q_2(x)+Q_3(x)+Q_4(x)+Q_5(x)
\ee
with
\begin{align*}
  Q_1(x):= \,& \delta(x_1)\delta(x_2-x_3)\delta(x_4-x_5)\,, \\  
  Q_2(x):= \,& \theta(x_3-x_2)\delta(x_1)\delta(x_2-x_5)\delta(x_3-x_4)\,, \\
  Q_3(x):= \,& \theta(x_1)\theta(x_3)\delta(x_5)\delta(x_1-x_2)\delta(x_3-x_4)\,, \\
  Q_4(x):= \,& \theta(x_1)\delta(x_3)\delta(x_1-x_2)\delta(x_4-x_5)\,, \\
  Q_5(x):= \,& \theta(x_1)\theta(x_2-x_1)\delta(x_5)\delta(x_1-x_4)\delta(x_2-x_3)\,.
\end{align*}

The contributions from each region are:
\begin{align}\label{qmet}
A^{\phi^4: (1)}_8 = \,&  \frac{1}{t_{123}t_{456}}+\frac{1}{t_{456}t_{781}}+\frac{1}{t_{781}t_{234}}+\frac{1}{t_{234}t_{678}}+\frac{1}{t_{678}t_{123}}\,, \nonumber \\
A^{\phi^4: (2)}_8 = \,&  \frac{1}{t_{567}}\left(\frac{1}{t_{123}}+\frac{1}{t_{234}}\right)\,,\nonumber \\
A^{\phi^4: (3)}_8 = \,&  \frac{1}{t_{812}}\left(\frac{1}{t_{345}}+\frac{1}{t_{567}}\right)\,,\nonumber \\
A^{\phi^4: (4)}_8 = \,&  \frac{1}{t_{345}}\left(\frac{1}{t_{678}}+\frac{1}{t_{781}}\right)\,,\nonumber \\
A^{\phi^4: (5)}_8 = \,&  \frac{1}{t_{456}t_{812}}\,.
\end{align} 
The amplitude $A^{\phi^4}_8$ is the sum over all five contributions and gives rise to the familiar expression in terms of $12$ Feynman diagrams. 

\subsection{One Region for All $n$}\label{oneRG}

In the next section we provide a diagrammatic technique for finding all regions contributing to $A^{\phi^4}_n$. In this last example, we study the contribution from the analog to $R_1$ for all $n$. The region is defined in the following proposition.

\begin{prop}\label{Rone}
The function 
\be  
 H(x) = \sum_{a<b}^{\rm odd}f_{[a,b]}(x)\,,
\ee 
defined in \eqref{Hregions}, vanishes in the region 
\be  
R_1 = \{ x_0 = x_1,\, x_2=x_3,\, x_4=x_5,\ldots ,\, x_{n-4}=x_{n-3}\}  \,.
\ee
\end{prop}

\begin{proof}
Since $H(x)$ is the sum of non-negative functions, $f_{[a,b]}(x)$, we have to show that each such function vanishes on $R_1$. Using the definitions \eqref{defFint}, \eqref{tropMin}, one has 
\begin{align}\label{fefe}
   f_{[a,b]}(x) & = \min (x_{a-2},x_{a-1},\ldots ,x_{b-3})-\min (x_{a-2},x_{a-1},\ldots ,x_{b-3},x_{b-2})\\ \nonumber & -\min (x_{a-3},x_{a-2},\ldots ,x_{b-3}) +\min (x_{a-3},x_{a-2},\ldots ,x_{b-3},x_{b-2})\,.
\end{align}  
By definition, $H(x)$ only contains $ f_{[a,b]}(x)$ with $b-a\equiv 1\! \mod 2$. This means that on $R_1$, either $x_{a-3}=x_{a-2}$ or $x_{b-3}=x_{b-2}$. This is easily seen by considering two cases: If $a\in 2\mathbb{Z}$ then $b\in 2\mathbb{Z}+1$ and therefore $a-3\in 2\mathbb{Z}+1$, $b-3\in 2\mathbb{Z}$, and $x_{b-3}=x_{b-2}$ on $R_1$. The same can be repeated when $a\in 2\mathbb{Z}+1$ to conclude that $x_{a-3}=x_{a-2}$. 

Finally, note that if $x_{a-3}=x_{a-2}$ then the first and third terms in \eqref{fefe} cancel each other while the second and fourth do too. If $x_{b-3}=x_{b-2}$ then the first and second cancel while the third and fourth do too. 
\end{proof}

In order to evaluate the contribution from $R_1$ it is convenient to define the following combination of kinematic invariants,
\be  
r_{ab}:= s_{ab}+s_{a,b+1}+s_{a+1,b}+s_{a+1,b+1}\,.
\ee 
It is also useful to write $r_{ab}$ in terms of planar invariants using \eqref{fourTerms},
\be\label{cami} 
r_{ab} = -t_{[a,b-1]}+t_{[a,b+1]}+t_{[a+2,b-1]}-t_{[a+2,b+1]}\,.
\ee 
Note that if $a\in 2\mathbb{Z}$ and $b\in 2\mathbb{Z}+1$, then all four invariants in \eqref{cami} belong to the set of $\phi^4$ invariants. 

Restricting the tropical potential function \eqref{tropP} to $R_1$ one finds
\be\label{mica} 
F_n(x) = G_n(x) = \sum_{a=2}^{\rm even}\sum_{b=5}^{\rm odd}r_{ab}\min (x_{a-2},x_{a},\ldots ,x_{b-5},x_{b-3}) + \sum_{a=2}^{{\rm even}}t_{a,a+1,a+2}x_{a-2}\,.
\ee 
In sums labeled ``even'' (``odd'') the index only takes even (odd) values. The first equality is due to the fact that on $R_1$ the function $H_n(x)=0$.

The function $G_n(x)$ has exactly the structure of a tropical potential for $m_{n/2+1}(\mathbb{I},\mathbb{I})$ if the labels are identified as $x_a \to x_{a/2}$. This is well-defined since $a$ only takes even values in \eqref{mica}. Instead of using the mapping, we keep the original labels and write the tropical potential for $m_{n/2+1}(\mathbb{I},\mathbb{I})$ as
\be  
F_{n/2+1}^{\phi^3}(x_0,x_2,\ldots ,x_n):= \sum_{a=2}^{\rm even}\sum_{b=5}^{\rm odd}s_{a,b+1}\min (x_{a-2},x_{a},\ldots ,x_{b-5},x_{b-3}) + \sum_{a=2}^{{\rm even}}s_{a,a+2}x_{a-2}\,.
\ee 
Matching the coefficients gives the map of kinematic invariants, 
\be\label{map34} 
r_{ab}=s_{a,b+1}\,,\quad t_{a,a+1,a+2}=s_{a,a+2}\,, \quad a\in \{ 2,4,6\ldots \}\,,\,\, b\in \{5,7,9,\ldots \} \,.
\ee
We conclude that the contribution of region $R_1$ to $A^{\phi^4}_n$ is nothing but $m_{n/2+1}(\mathbb{I},\mathbb{I})$ with kinematic invariants given by \eqref{map34}.

This result prompts the following proposition.

\begin{prop}\label{ampProp}
Consider $A^{\phi^4}_n$ evaluated on the following subspace of kinematic invariants,
\be 
t_{[a,b]} =\frac{1}{\epsilon}\,, \quad a\in \{3,5,\ldots ,n-3\}\,, \,\, b\in\{ a+2,a+4,\ldots ,n-1\} \,.
\ee 
to produce a function $A^{\phi^4(\epsilon)}_n$.
Then,
\be 
\lim_{\epsilon\to 0}A^{\phi^4(\epsilon)}_n = m_{n/2+1}(\mathbb{I},\mathbb{I}) ,
\ee 
for some bijection of the set of planar kinematic invariants.
\end{prop}

\begin{proof}
Let us start by proving the spaces of kinematic invariants possess the same cardinality. The set of planar kinematic invariants of $A^{\phi^4}_n$ has cardinality of $n(n-4)/4$. For example, it is $0,3,8,15$ for $n=4,6,8,10$. In the statement of the proposition, $(n-2)(n-4)/8$ of the kinematic invariants are set to $1/\epsilon$ so there are only $(n-4)(n+2)/8$ left. Let us introduce $n_3:=n/2+1$, then $(n-4)(n+2)/8 = n_3(n_3-3)/2$ which is the cardinality of the set of planar kinematic invariants for $m_{n_3 = n/2+1}(\mathbb{I},\mathbb{I})$. 

The rest of the proof is based on the fact that in the kinematic space of interest and the corresponding limit, only region $R_1$ contributes to the amplitude. Since the result from region $R_1$ was shown to be $m_{n_3 = n/2+1}(\mathbb{I},\mathbb{I})$ for some bijetion of the kinematic invariants in this subsection, then the statement of the proposition follows. In order to actually complete the proof, we first need to classify all regions and this is done in the next section. We therefore postpone the completion of the proof to the end of the next section in subsection \ref{combiProof}.

\end{proof}

\section{Combinatorial Description of Regions}\label{combi}

In the previous section we provided some examples of how the global Schwinger formula for $A^{\phi^4}_n$ is evaluated. The result decomposes as a sum over regions (cones) which are in bijection with noncrossing chord diagrams. In this section we provide a systematic study of the structure of the regions. 

The unexpected appearance of $m_{n/2+1}(\mathbb{I},\mathbb{I})$ in the contribution from one of the regions, explained in section \ref{oneRG}, motivates a similar interpretation for the other regions. Indeed, in all examples we have studied we find that all contributions are related to $m_{n/2+1}(\alpha ,\mathbb{I})$ for some choice of ordering $\alpha$. 

In order to make the study systematic, we propose a diagrammatic procedure for finding all the regions that contribute to $A^{\phi^4}_n$ and show how each such region is in bijection with a cubic $m_{n/2+1}(\alpha ,\mathbb{I})$ amplitude.

\subsection{Regions for $A_n^{\phi^4}$: Non-Crossing Chord Diagrams}

Let us start by defining non-crossing chord diagrams in our context. 

\begin{defn}\label{chords}

Place $n-2$ points labeled $0,1, \ldots ,n-3$ in increasing order on the real line. A {\it non-crossing chord} diagram is a perfect matching of the points such that all edges can be drawn as chords on the upper half plane without any crossings. Let us denote the chord connecting points $a$ and $b$ as $\theta_{ab}$.  
\end{defn}
	
\begin{conj}\label{conj1}
The regions contributing to $A^{\phi^4}_n$ are in bijection with the set of all $\textrm{C}_{n/2-1}$ possible non-crossing chord diagrams defined in \ref{chords}. Moreover, the region $R$ corresponding to a particular diagram is obtained as follows:
\begin{itemize}
    \item For each chord $\theta_{ab}$ set $x_a=x_b$.
    \item If a chord $\theta_{ab}$ surrounds another chord $\theta_{cd}$, then $x_a=x_b < x_c=x_d$.
\end{itemize}
In other words, the regions defined by non-crossing chord diagrams are all the solutions to $H(x)=0$.
\end{conj}

\vskip0.1in

Let us note that the case in which no chord surrounds any other chord corresponds to 
\be 
R = \{ x_0 =x_1,\, x_2=x_3,\, \ldots ,\, x_{n-4}=x_{n-3} \} \,. 
\ee
This is nothing but the region $R_1$ which was proven to set $H(x)=0$ in Proposition \eqref{Rone}. 

\begin{example}

Consider two of the examples presented in section \ref{examples4}. For $n=4$ there is a single chord diagram. It has a single chord $\theta_{01}$ and therefore the region is given by $x_0=x_1$. Recall that $x_0=0$ and so $x_1=0$. 

For $n=6$ there are two non-crossing chord diagrams as shown in figure \ref{fig:n6_regions4}.

\begin{figure}[h!]
	\centering
	\includegraphics[width=0.90\linewidth]{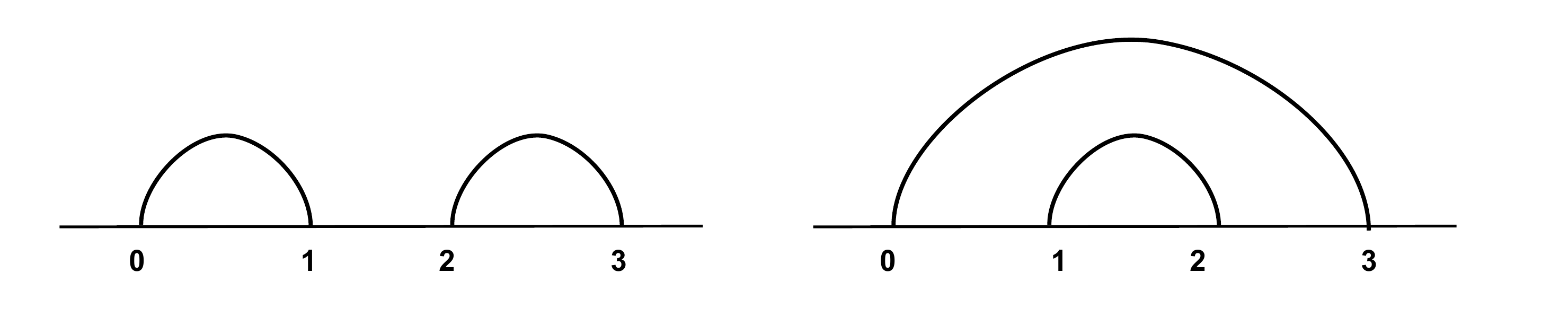}
	\caption{Non-crossing chord diagrams for $n=6$. On the right, the chord $\theta_{03}$ surrounds the chord $\theta_{12}$ and therefore the condition $x_0<x_1$ is imposed.}
	\label{fig:n6_regions4}
\end{figure}

The corresponding regions can be seen to match $R_1$ and $R_2$ in \eqref{region6}, i.e.,
\be
R_1 = \{ x_0 = x_1,\, x_2 = x_3 \}\,, \quad R_2 = \{ x_0=x_3,\, x_1=x_2,\, x_0 < x_1 \}\,. 
\ee 

Finally, we leave as an exercise to the reader to check that the five regions for $n=8$ presented in \eqref{eight4} correspond to the diagrams in figure \ref{fig:n8_regions4}.

\begin{figure}[h!]
	\centering
	\includegraphics[width=0.95\linewidth]{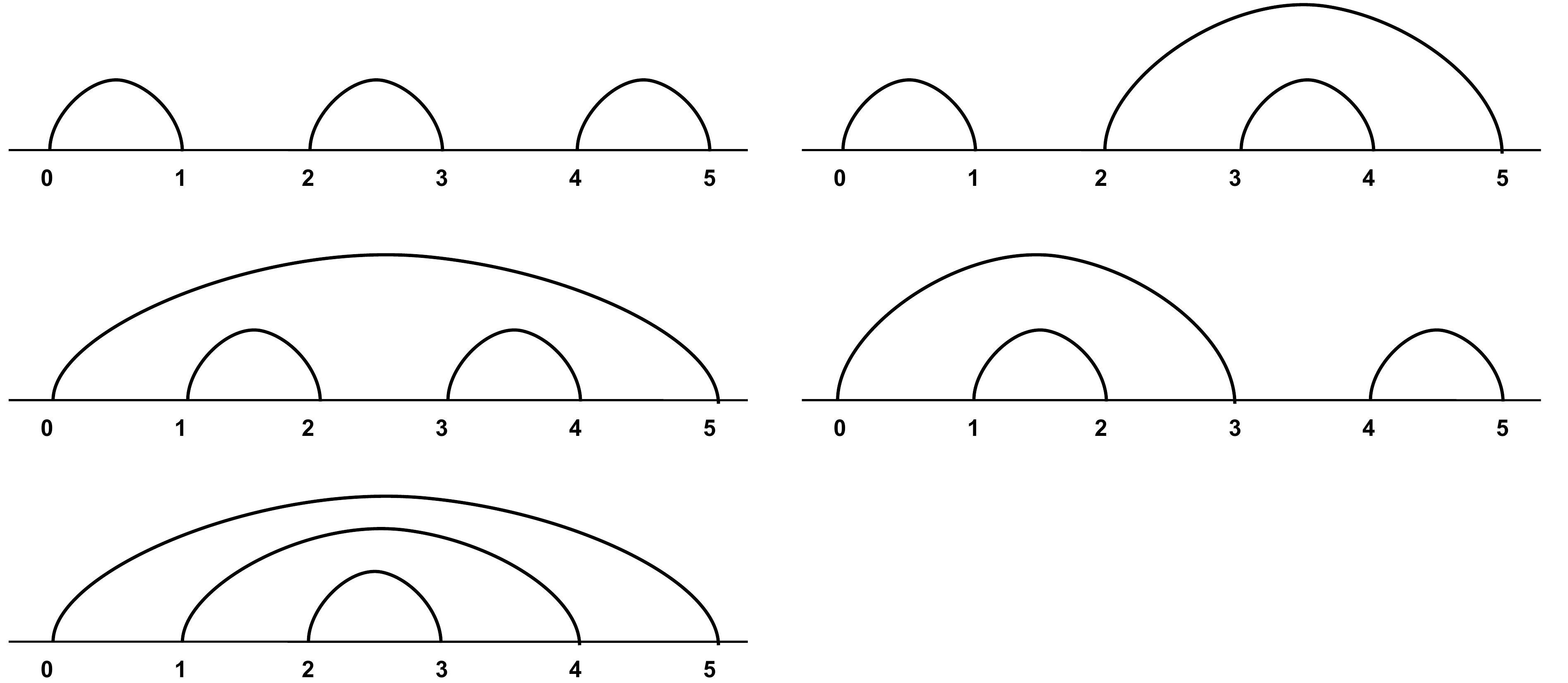}
	\caption{Non-crossing chord diagrams for $n=8$. In the second diagram $\theta_{25}$ surrounds $\theta_{34}$ and therefore $x_2<x_3$. In the third diagram $\theta_{05}$ surrounds both $\theta_{12}$ and $\theta_{34}$ and therefore $x_0<x_1$ and $x_0<x_3$. In the fourth diagram $\theta_{03}$ surrounds $\theta_{12}$ so $x_0<x_1$. Finally, in the fifth diagram $\theta_{05}$ surrounds $\theta_{14}$ which surrounds $\theta_{23}$ so $x_0<x_1<x_2$.}
	\label{fig:n8_regions4}
\end{figure}

\end{example}

In the last example of section \ref{examples4} we found that the contribution to region $R_1$ was computed by a biadjoint $\phi^3$ amplitude with $n/2+1$ particles, i.e., $m_{n/2+1}(\mathbb{I},\mathbb{I})$. The attentive reader might have noticed that in all examples provided so far, the structure of the answer resembles that of $m_{n/2+1}(\alpha ,\mathbb{I})$ for some permutation $\alpha$. We leave the precise connection between $\alpha$ and a region for future work and here we concentrate on the schematic structure of $A^{\phi^4}_n$ for which we have an all $n$ proposal.

\subsection{Connecting Regions to Products of $\phi^3$ Amplitudes: Towards $m_{n/2+1}(\alpha, \mathbb{I})$}\label{conMi}

In order to understand the structure of each region, it is useful to introduce an additional chord to the non-crossing chord diagrams described above. More precisely, we introduce two new points, which could be denoted $-1$ and $n-2$, and we always draw a chord between them. The point $-1$ is located to the left of $0$ and $n-2$ is to the right of $n-3$ so that the chord $\theta_{-1,n-2}$ surrounds the whole diagram. This can be understood as a way of introducing into the figure the fixed particles $1$ and $2$ in the parameterization \eqref{preT}.  

\begin{defn}
An extended non-crossing chord diagram associated to $A^{\phi^4}_n$ is a non-crossing chord diagram on $n$ points labeled by $\{ -1,0,1,2,\ldots ,n-3,n-2 \}$ in which $\theta_{-1,n-2}$ is always included. We also define a meadow of an extended non-crossing chord diagram as any region in the diagram delimited by more than one chord and by the line where the points lie. 
\end{defn}

The claim is that a meadow delimited by $m$ chords and the real line corresponds to a biadjoint ($m+1$)-subamplitude participating in $m_{n/2+1}(\alpha ,\mathbb{I})$. Moreover, we also claim that any chord $\theta_{ab}$ shared by {\it two} meadows corresponds to a propagator in $m_{n/2+1}(\alpha ,\mathbb{I})$ of the form $1/t_{[a+3,b+2]}$. This also fixes the topology of the cubic double-ordered amplitude.

Before describing the consequences of this proposal, let us give some examples to illustrate it.

\begin{example}

Consider the region described by the diagram in figure \ref{n10eg} for $n=10$.

Using the diagram it is easy to recognize the region as
\be 
R = \{ x_0=x_1,\, x_2=x_5, \, x_3=x_4,\, x_6=x_7,\, x_2<x_3 \} \,.
\ee 
Since the green meadow is delimited by 4 chords and the real line then it corresponds to a 5-particle subamplitude of $m_6(\alpha,\mathbb{I})$, while the blue meadow is delimited by 2 chords and the real line and thus corresponds to a 3-particle subamplitude.

The disk diagram on the right is intended to represent the topology of $m_6(\alpha,\mathbb{I})$ given the diagram on the left. The reader familiar with the CHY description of biadjoint partial amplitudes would recognize the disk diagram as encoding the two orderings $\alpha$ and $\mathbb{I}$.   

\begin{figure}[H]
\includegraphics[width=15cm]{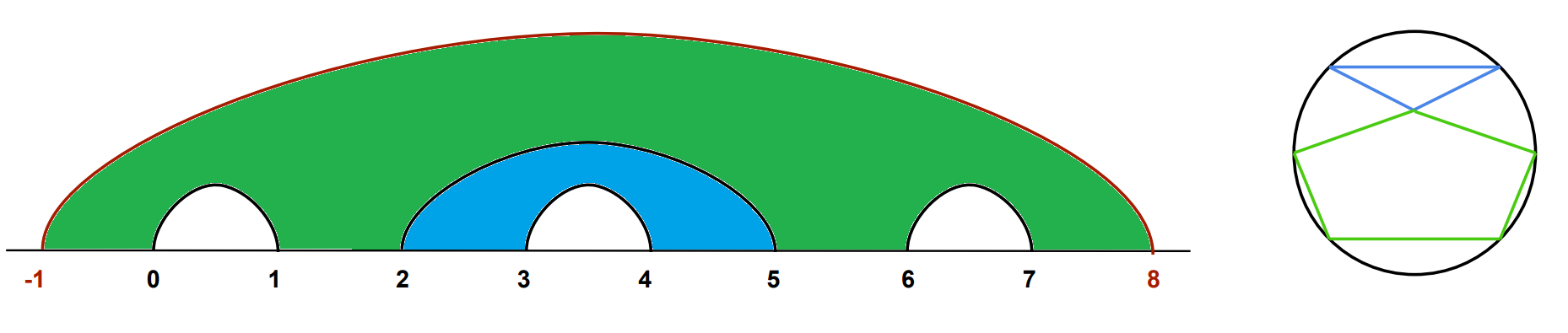}
\centering
\caption{\textit{Left:} An extended non-crossing chord diagram for $n=10$ where the meadows have been coloured. The additional points $-1$ and $8$ together with the chord $\theta_{-1,8}$ joining them are coloured in magenta. \textit{Right:} Disc diagram of an $m_6(\alpha,\mathbb{I})$ amplitude corresponding to the contribution of the region on the left.}
\label{n10eg}
\end{figure}

Finally, notice that the chord $\theta_{25}$ is shared by two meadows, hence it generates the propagator $1/t_{567}$ in $m_6(\alpha,\mathbb{I})$. The conclusion is that the contribution of this region to $A_{10}^{\rm \phi^4}$ is schematically given by
\be 
m_3\times m_5\times \frac{1}{t_{567}}\,.
\ee 

\end{example}

\begin{example}
Consider another region contributing to $A_{10}^{\rm \phi^4}$, described by the diagram in figure \ref{n10egTwo}.
\begin{figure}[H]
\includegraphics[width=15cm]{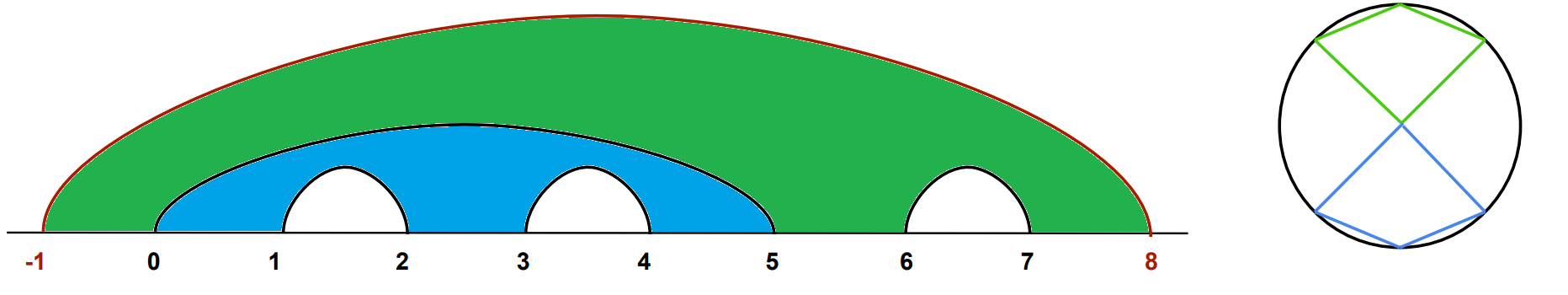}
\centering
\caption{\textit{Left:} An extended non-crossing chord diagram for $n=10$ where the meadows have been coloured. \textit{Right:} Disc diagram of an $m_6(\alpha,\mathbb{I})$ amplitude corresponding to the contribution of the region on the left.}
\label{n10egTwo}
\end{figure}

In this case we have two 4-particle subamplitudes and one propagator of the form $1/t_{[3,7]}$. The contribution of this region to $A_{10}^{\rm \phi^4}$ is schematically given by
\be 
(m_4)^2\times \frac{1}{t_{[3,7]}}\,.
\ee 
Note that $(m_4)^2$ stands for the product of two distinct four-point $\phi^3$ amplitudes. Since we are only interested in the schematic structure, i.e. in the number of amplitudes of a given type, we keep track of that using exponents.
\end{example}

\begin{example}
Consider now a region contributing to $A_{8}^{\rm \phi^4}$. The region is defined by the diagram in figure \ref{n8egR5}. We leave as an exercise for the reader to show that this corresponds to region $R_5$ in the example given in section \ref{sec:Eight_Point_Amplitude}.
\begin{figure}[H]
\includegraphics[width=15cm]{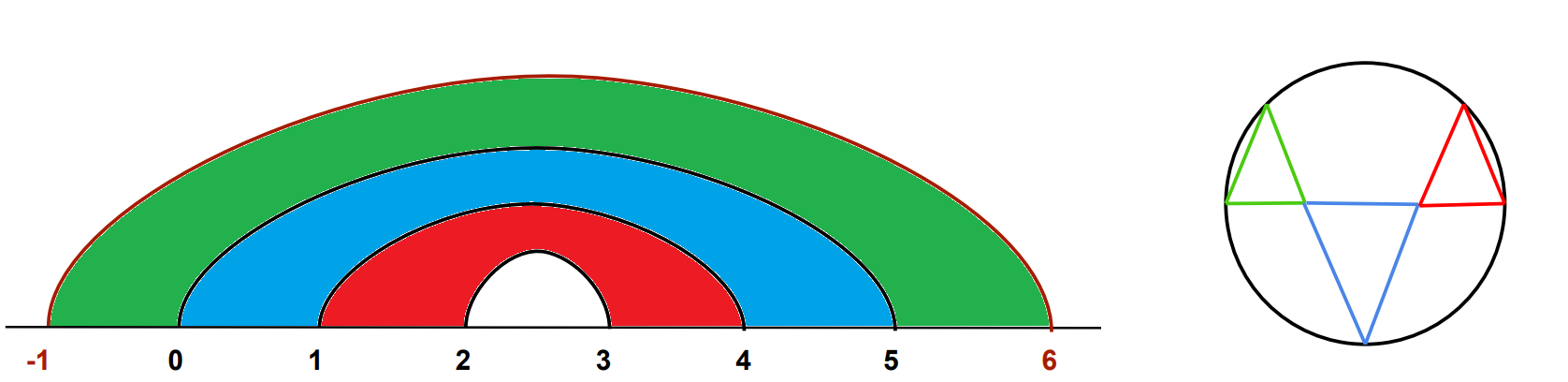}
\centering
\caption{\textit{Left:} An extended non-crossing chord diagram for $n=8$ where the meadows have been coloured. \textit{Right:} Disc diagram of an $m_5(\alpha,\mathbb{I})$ amplitude corresponding to the contribution of the region on the left.}
\label{n8egR5}
\end{figure}
The extended non-crossing chord diagram contains three meadows delimited by two chords (i.e. three 3-particle subamplitudes) and two propagators corresponding to the chords $\theta_{05}$ and $\theta_{14}$. Accoding to our proposal, these propagators are $1/t_{[3,7]}=1/t_{812}$ and $1/t_{456}$, respectively. Once again, the schematic form of the contribution is
\be 
(m_3)^3\times \frac{1}{t_{812}}\times \frac{1}{t_{456}}\,.
\ee 
If we define $m_3:=1$ this is exactly the contribution  $A_8^{\phi^4:(5)}$ presented in \eqref{qmet}.

\end{example}

Based on these and many other examples, we have found a formula that reproduces the schematic structure of $A_{n}^{\rm \phi^4}$ in every case. The formula is based on the Lagrange inversion procedure relating the series expansion of a function $f(x)$ with that of its compositional inverse. We review some related material in appendix \ref{lagrange}. Here we simply present the final form of the proposal. Let
\be 
h(x) =\sum_{i=0}^\infty h_{i} x^i := \sum_{i=0}^\infty m_{i+2} x^i\,,
\ee 
where $m_{i+2}$ represents a generic $(i+2)$-particle amplitude in the biadjoint $\phi^3$ scalar theory of the form $m_{i+2}(\mathbb{I},\mathbb{I})$. Since the mass dimension of $m_{i+2}(\mathbb{I},\mathbb{I})$ is $-2(i-1)$ we are motivated to define $m_2:=P^2$ and $m_3:=1$. Here $1/P^2$ represents a generic propagator. 

The claim is that the schematic form of the amplitude $A^{\phi^4}_n$ is given by
\be\label{propi} 
A^{\phi^4}_n = \left(\frac{2}{n\, h_0^{n/2-1}}\right)\frac{1}{2\pi i}\oint_{|z|=\epsilon}dz\left(\frac{h(z)}{z}\right)^{n/2}\,. 
\ee 

Let us compute the first few cases of \eqref{propi},
\begin{align*}
A^{\phi^4}_4 = &\, h_1  = m_3 \,, \\
A^{\phi^4}_6 = &\, \frac{h_0^2h_2+h_0h_1^2}{h_0^2} = m_4+m_3^2\frac{1}{P^2}\,, \\
A^{\phi^4}_8 = &\, \frac{h_3 h_0^3+3 h_1 h_2 h_0^2+h_1^3 h_0}{h_0^3} = m_5+3m_3m_4\frac{1}{P^2}+m_3^3\left(\frac{1}{P^2}\right)^2\,, \\
A^{\phi^4}_{10} = &\, m_6+ 4 m_5
   m_3\frac{1}{P^2}+2 m_4^2\frac{1}{P^2}+6 m_4 m_3^2\left(\frac{1}{P^2}\right)^2+m_3^4\left(\frac{1}{P^2}\right)^3\,.
\end{align*}

There are several consistency checks that can be done on \eqref{propi}. 

The first is that the number of non-crossing chord diagrams with $n/2-1$ chords is $\textrm{C}_{n/2-1}$. Therefore if we set all $m_r:=1$ so that the contribution from each region is unity, one must find that $A^{\phi^4}_n = \textrm{C}_{n/2-1}$.

This means that we must set 
\be 
h(x) = \sum_{i=1}^\infty x^i = \frac{1}{1-x}\,,
\ee 
and evaluate
\be 
A^{\phi^4}_n =\left(\frac{2}{n}\right)\frac{1}{2\pi i}\oint_{|z|=\epsilon}dz\left(\frac{1}{z(1-z)}\right)^{n/2} = \textrm{C}_{n/2-1}\,. 
\ee 
The last equality follows from the Lagrange inversion formula with $f(x)=x(1-x)$ and $g(x)=x B_2(x)$, where $B_2(x)$ is the generating function of Catalan numbers. 

The second check is that if $A^{\phi^4}_n$ is evaluated on ``planar kinematics'' \cite{Cachazo:2013iea,Early:2019eun,Cachazo:2020wgu}, i.e. on the kinematic point where all planar Mandelstam invariants that participate in $A^{\phi^4}_n$ are unity, $t_{[a,b]}=1$, then $A^{\phi^4}_n$ simply counts the number of planar ternary trees (with all internal vertices of degree four). The numbers are known to be given by the Fuss-Catalan sequence, ${\rm FC}_{n/2-1}(3,1)$. For $n=4,6,8,10$ one has  ${\rm FC}_{n/2-1}(3,1)=1,3,12,55$. This check can be done by realizing that on planar kinematics $m_n=\textrm{C}_{n-2}$ and therefore
\be 
h(x) = \sum_{i=1}^\infty \textrm{C}_i\, x^i =B_2(x) =\frac{1-\sqrt{1-4x}}{2x}\,.
\ee 
As shown in appendix A in \eqref{threeProof}, it is indeed the case that 
\be
{\rm FC}_r(3,1) = \frac{1}{2\pi i}\oint_{|z|=\epsilon} \frac{dz}{r+1}\left(\frac{1-\sqrt{1-4z}}{2z^2}\right)^{r+1}\,,
\ee 
which gives the required relation when $r=n/2-1$.

\subsection{Completing the Proof of Proposition \ref{ampProp}}\label{combiProof}

In order to complete the proof of Proposition \ref{ampProp} we have to show that all regions that contribute to $A^{\phi^4}_n$, except for $R_1$, are ${\cal O}(\epsilon)$ when  
\be 
t_{[a,b]} =\frac{1}{\epsilon}\,, \quad a\in \{3,5,\ldots ,n-3\}\,, \,\, b\in\{ a+2,a+4,\ldots ,n-1\} \,.
\ee 
Recall that $R_1$ is the region corresponding to $n/2-1$ non-crossing chords so that none is surrounded by any other. According to the rules explained in this section, this means that no propagator is generated. One the other hand, every single other region has at least one chord surrounded by another, say $\theta_{ef}$ and therefore there is at least one propagator in the region's contribution to the amplitude. The propagator is $1/t_{[e+3,f+2]}$. Clearly $f-e\geq 3$ so that the chord can contain at least another one. This means that the chords of interest can only have $e\in \{ 0,1,\ldots ,n-6 \}$ and $f\in \{ e+3,e+4,\ldots ,n-3\}$. Therefore each region different from $R_1$ contains at last one propagator of the form $1/t_{[a,b]}$ with $a\in \{ 3,4,\ldots ,n-3 \}$ and $b\in \{ a+2,a+3,\ldots ,n-1\}$. But this is exactly the range of propagators set to $\epsilon$ and this concludes the proof.

\section{From $\phi^3$ Amplitudes to $\phi^p$ Amplitudes}\label{phip}

In this section we extend the limiting procedure used to obtain $A_n^{\phi^4}$ from $m_n(\mathbb{I},\mathbb{I})$ to make a general conjecture for any $A_n^{\phi^p}$ amplitude and its global Schwinger formulation. We also propose a diagrammatic procedure for finding all the regions that contribute to $A_n^{\phi^p}$ and point out a connection with  $m_{(n+2(p-3))/(p-2)}(\alpha,\mathbb{I})$ amplitudes.

To start with, the limiting procedure that generates $A_n^{\phi^p}$ from $m_n(\mathbb{I},\mathbb{I})$ is the following.
\begin{prop}\label{propp}
Consider the region of the kinematic space of $n$ massless particles where $t_{[a,b]}= 1/\epsilon$ whenever $b-a\not\equiv 0\!\! \mod p-2$ and let $m_n^{(\epsilon)}(\mathbb{I},\mathbb{I})$ denote $m_n(\mathbb{I},\mathbb{I})$ evaluated on it. Then 
\be 
A^{\phi^p}_n = \lim_{\epsilon\to 0}\frac{1}{(\textrm{C}_{p-2}\,\epsilon^{p-3})^{\frac{n-2}{p-2}}}m_n^{(\epsilon)}(\mathbb{I},\mathbb{I}) ,\label{aplimit}
\ee 
where $\textrm{C}_m$ is the $m^{th}$ Catalan number.
\end{prop}

\begin{proof}
The proof is analogous to that of Proposition \ref{qeri}. First, consider any Feynman diagram ${\cal T}^{(p)}$ of $A^{\phi^p}_n$, that is, any $(p-1)$-ary planar tree on $n$-leaves. Such a diagram has $(n-2)/(p-2)$ vertices of degree $p$. The strategy is again to find out how many planar binary trees give rise to ${\cal T}^{(p)}$ by collapsing edges, i.e., taking their length to zero. This is easily done by realizing that for each degree-$p$ vertex of ${\cal T}^{(p)}$ there are exactly $\textrm{C}_{p-2}$ ways, compatible with planarity, of growing a tree to produce $\textrm{C}_{p-2}$ degree-three vertices. This means that there are $\textrm{C}_{p-2}^{(n-2)/(p-2)}$ binary trees that give rise to ${\cal T}^{(p)}$. Of course, not all binary trees descend to a $(p-1)$-ary diagram. Note that under the kinematics in the proposition, $\phi^3$ Feynman diagrams that collapse to $\phi^p$ diagrams have exactly $(n-p)/(p-2)$ $\epsilon$-independent propagators and therefore $(n-3)-(n-p)/(p-2)=(p-3)(n-2)/(p-2)$ which become $1/t = \epsilon$. Diagrams that do not produce a $\phi^p$ diagram have at least one extra propagator of the form $1/t=\epsilon$. Therefore, in the limit $\epsilon\to 0 $ the $\phi^p$ amplitude is recovered.
\end{proof}

As in section \ref{sD}, one can write the global Schwinger formula for $A_n^{\phi^p}$ as a single integral 
\be
A_n^{\phi^p}=\int_{\mathbb{R}^{n-3}}d^{n-3}x\,\textrm{exp}\left(-\sum_{a<b}^{{\cal K}_p}t_{[a,b]}f_{[a,b]}(x)\right)Q(x)\,,\label{GSp}
\ee
where ${\cal K}_p$ means that the sum is over ordered pairs $(a,b)$ such that $b-a\equiv 0\!\! \mod p-2$. Here $Q(x)$ is defined as

\be
Q(x):=\lim_{\epsilon\to 0^+}\frac{1}{(\textrm{C}_{p-2}\,\epsilon^{p-3})^{\frac{n-2}{p-2}}}\,\textrm{exp}\left(-\frac{1}{\epsilon}H(x)\right)
\ee
with
\be
H(x)=\sum_{a<b: (a,b)\notin {\cal K}_p}f_{[a,b]}(x)\, .\label{Hp}
\ee
Note that due to the non-negativity of $H(x)$, $Q(x)$ only has support in regions where $H(x)=0$. Again, the distribution $Q(x)$ becomes a sum over distributions that localize the integral to these regions. This means that equation \eqref{GSp} can also be understood as a sum over regions, where these as associated to diagrams as explained in the next subsection.

\subsection{Combinatorial Description of Regions}

In this subsection we conjecture that the solutions of $H(x)=0$ are regions of dimension $n/(p-2)-1$ in $\mathbb{R}^{n-3}$ which are classified by non-crossing $(p-2)$-chord diagrams. The definition of non-crossing $(p-2)$-chord diagrams, in our context, is the following.

\begin{defn}\label{kchords}

Place $n-2$ points labeled $0,1, \ldots ,n-3$ on the real line in increasing order. A {\it non-crossing $(p-2)$-chord} diagram is a perfect matching of the points such that each matching involves $(p-2)$ points joined by a $(p-2)$-chord and drawn on the upper half plane without any crossings. Let us denote the $(p-2)$-chord connecting points $a_1,a_2, \ldots ,a_{p-2}$ as $\theta_{a_1,a_2,...,a_{p-2}}$ (for general k-chord diagrams see e.g. \cite{kChords}.) 
\end{defn}
	
\begin{conj}\label{conjkchords}
The regions contributing to $A^{\phi^p}_n$ are in bijection with the set of all $\textrm{FC}_{(n-2)/(p-2)}(p-2,1)$\footnote{Recall that $\textrm{FC}_m(q,r)$ is the Fuss-Catalan number given by $$\textrm{FC}_m(q,r)\equiv \frac{r}{mq+r}{mq+r \choose m}\,.$$ Note that for $q=2$ and $r=1$ the Fuss-Catalan numbers coincide with the Catalan numbers, i.e. $\textrm{FC}_m(2,1)=\textrm{C}_m$.} possible non-crossing $(p-2)$-chord diagrams. Moreover, the region $R$ corresponding to a particular diagram is obtained as follows:
\begin{itemize}
    \item For each $(p-2)$-chord $\theta_{a_1,a_2,...,a_{p-2}}$ set $x_{a_1}=x_{a_2}=\cdots=x_{a_{p-2}}$.
    \item If a $(p-2)$-chord $\theta_{a_1,a_2,...,a_{p-2}}$ surrounds another $(p-2)$-chord $\theta_{b_1,b_2,...,b_{p-2}}$, then $x_{a_1}=x_{a_2}=\cdots=x_{a_{p-2}} < x_{b_1}=x_{b_2}=\cdots=x_{b_{p-2}}$.
\end{itemize}
In other words, the regions defined by the non-crossing $(p-2)$-chord diagrams are all the solutions to $H(x)=0$, where $H(x)$ is given by \eqref{Hp}, and the sum of their contributions produces all the $\textrm{FC}_{(n-2)/(p-2)}(p-1,1)$ trees of $\phi^p$.
\end{conj}

\begin{example}
Consider the $n=10$ amplitude for $\phi^6$. There are four non-crossing $4$-chord diagrams and are shown in figure \ref{n10p6}.

\begin{figure}[H]
\includegraphics[width=10cm]{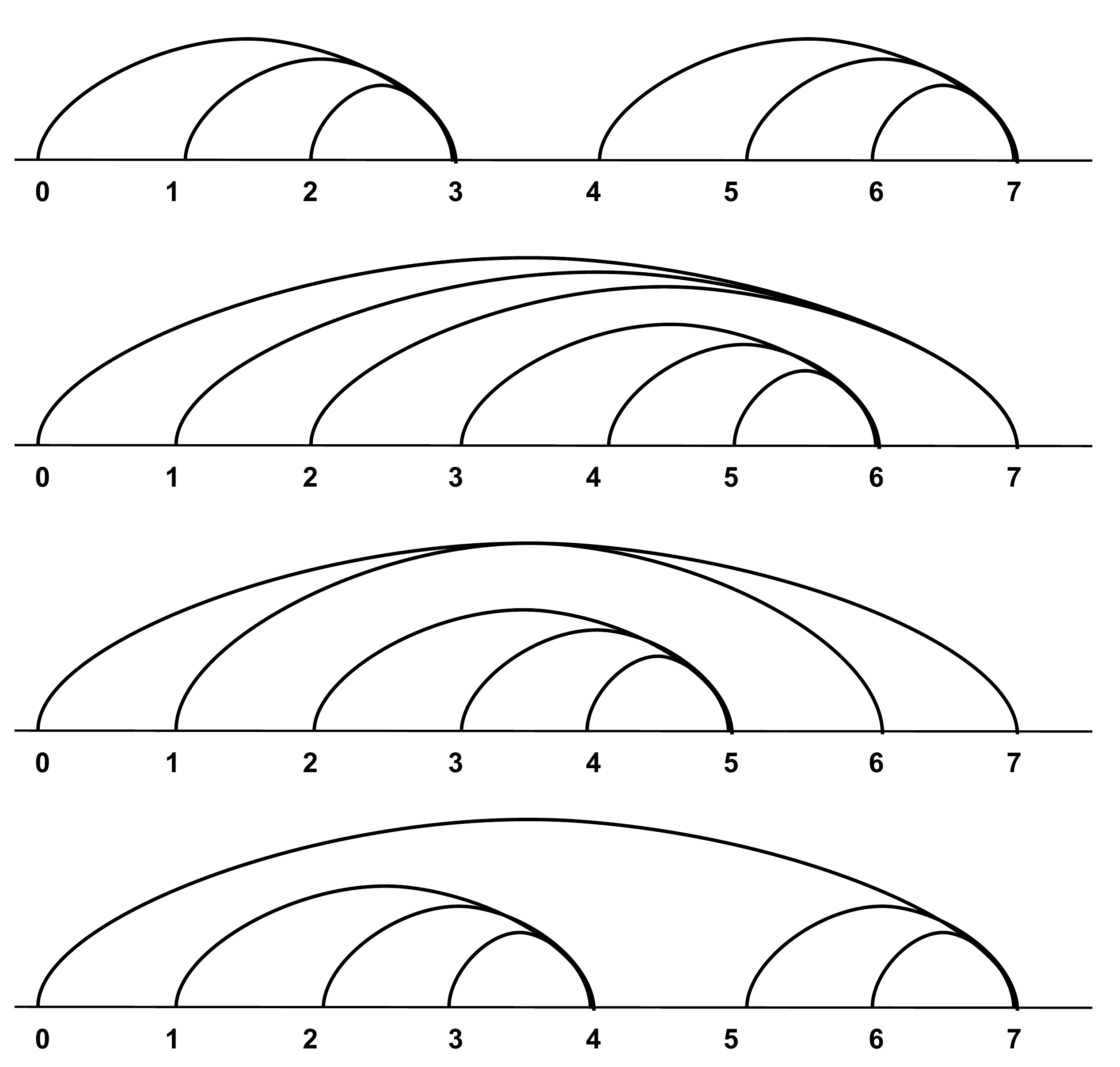}
\centering
\caption{All possible non-crossing 4-chord diagrams for $n=10$ and $p=6$. Each diagram contains exactly two 4-chords. In the top diagram one 4-chord joins points $0,1,2,3$ while the second 4-chord joins $4,5,6,7$. In the second diagram the 4-chord $\theta_{0127}$ surrounds $\theta_{3456}$. In the third, $\theta_{0167}$ surrounds $\theta_{2345}$. In the last diagram, $\theta_{0567}$ surrounds $\theta_{1234}$.}
\label{n10p6}
\end{figure}

Reading from top to bottom and recalling that $x_0=0$, the four regions generated by these diagrams correspond, respectively, to

\begin{align*}
R_1= & \, \{x_1=x_2=x_3=0,\, x_4=x_5=x_6=x_7\}\,, \\
R_2= & \, \{x_1=x_2=x_7=0,\, x_3=x_4=x_5=x_6,\, x_3>0\}\,, \\
R_3= & \, \{x_1=x_6=x_7=0,\, x_2=x_3=x_4=x_5,\, x_2>0\}\,, \\
R_4= & \, \{x_5=x_6=x_7=0,\, x_1=x_2=x_3=x_4,\, x_1>0\}\,. 
\end{align*}
Therefore, the distribution $Q(x)$ is given by 
$$Q(x)=Q_1(x)+Q_2(x)+Q_3(x)+Q_4(x)$$
with
\begin{align*}
Q_1:= & \, \delta(x_1)\delta(x_2)\delta(x_3)\delta(x_4-x_5)\delta(x_5-x_6)\delta(x_6-x_7)\,, \\
Q_2:= & \, \theta(x_3)\delta(x_1)\delta(x_2)\delta(x_7)\delta(x_3-x_4)\delta(x_4-x_5)\delta(x_5-x_6)\,, \\
Q_3:= & \, \theta(x_2)\delta(x_1)\delta(x_6)\delta(x_7)\delta(x_2-x_3)\delta(x_3-x_4)\delta(x_4-x_5)\,, \\
Q_4:= & \, \theta(x_1)\delta(x_5)\delta(x_6)\delta(x_7)\delta(x_1-x_2)\delta(x_2-x_3)\delta(x_3-x_4)\,.
\end{align*}
The contributions from each region are
\begin{align}
\nonumber A_{10}^{\phi^6:(1)} = & \, \frac{1}{t_{[1,5]}}+\frac{1}{t_{[2,6]}}\,, \hspace{5mm} A_{10}^{\phi^6:(2)} =  \, \frac{1}{t_{[5,9]}}\,,  \\
A_{10}^{\phi^6:(3)} = & \, \frac{1}{t_{[4,8]}}\,, \hspace{18
mm} A_{10}^{\phi^6:(4)} =  \, \frac{1}{t_{[3,7]}}\,.
\end{align}
The amplitude $A_{10}^{\phi^6}$ is the sum over all $\textrm{FC}_2(4,1)=4$ contributions and gives rise to the familiar expression with $\textrm{FC}_2(5,1)=5$ Feynman diagrams.

\end{example}

\begin{example}
Consider now the $n=11$ case for $\phi^5$. There are 12 non-crossing $3$-chord diagrams and are represented in figure \ref{n11p5}.

\begin{figure}[H]
\includegraphics[width=15cm]{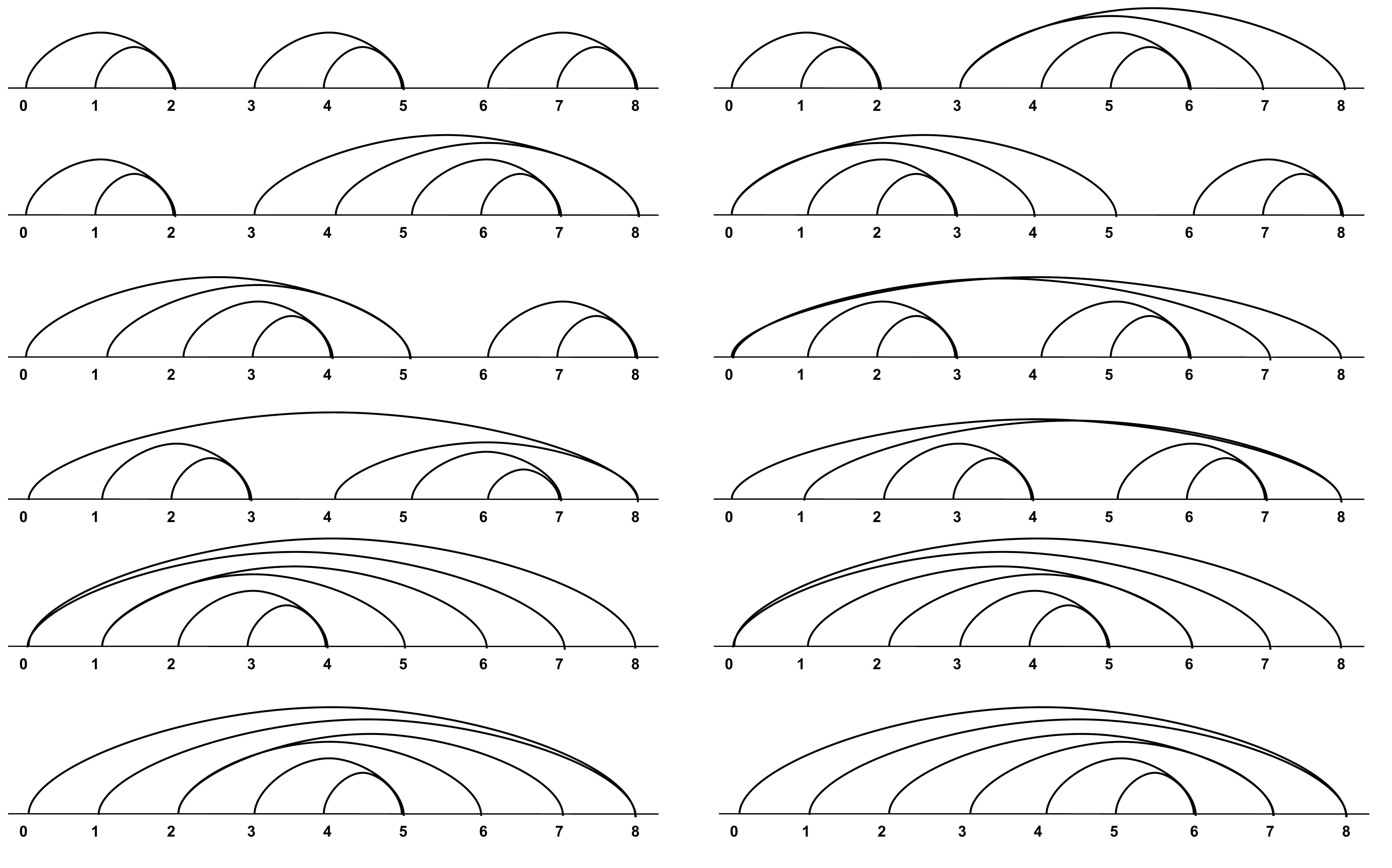}
\centering
\caption{All possible non-crossing 3-chord diagrams for $n=11$ and $p=5$.}
\label{n11p5}
\end{figure}

Reading from left to right and top to bottom, the 12 regions generated by these diagrams correspond, respectively, to

\begin{align*}
R_1= & \, \{x_1=x_2=0,\, x_3=x_4=x_5,\, x_6=x_7=x_8\}\,, \\
R_2= & \, \{x_1=x_2=0,\, x_3=x_7=x_8,\, x_4=x_5=x_6,\, x_3<x_4\}\,, \\
R_3= & \, \{x_1=x_2=0,\, x_5=x_6=x_7,\, x_3=x_4=x_8,\, x_3<x_5\}\,, \\
R_4= & \, \{x_4=x_5=0,\, x_1=x_2=x_3,\, x_6=x_7=x_8,\, x_1>0\}\,, \\
R_5= & \, \{x_1=x_5=0,\, x_2=x_3=x_4,\, x_6=x_7=x_8,\, x_2>0\}\,, \\
R_6= & \, \{x_7=x_8=0,\, x_1=x_2=x_3,\, x_4=x_5=x_6,\, x_1>0,\, x_4>0\}\,, \\
R_7= & \, \{x_4=x_8=0,\, x_1=x_2=x_3,\, x_5=x_6=x_7,\, x_1>0,\, x_5>0\}\,, \\
R_8= & \, \{x_1=x_8=0,\, x_2=x_3=x_4,\, x_5=x_6=x_7,\, x_2>0,\, x_5>0\}\,, \\
R_9= & \, \{x_7=x_8=0,\, x_1=x_5=x_6,\, x_2=x_3=x_4,\, x_2>x_1>0\}\,, \\
R_{10}= & \, \{x_7=x_8=0,\, x_1=x_2=x_6,\, x_3=x_4=x_5,\, x_3>x_1>0\}\,, \\
R_{11}= & \, \{x_1=x_8=0,\, x_2=x_6=x_7,\, x_3=x_4=x_5,\, x_3>x_2>0\}\,, \\
R_{12}= & \, \{x_1=x_8=0,\, x_2=x_3=x_7,\, x_4=x_5=x_6,\, x_4>x_2>0\}\,.
\end{align*}
We leave as an exercise to the reader to find the distributions associated to these regions and to show that the contributions from each region are:
\begin{align}
\nonumber A_{11}^{\phi^5:(1)} = & \, \frac{1}{t_{[2,5]}t_{[8,11]}}+\frac{1}{t_{[2,5]}t_{[2,8]}}+\frac{1}{t_{[5,8]}t_{[2,8]}}+\frac{1}{t_{[5,8]}t_{[5,11]}}+\frac{1}{t_{[8,11]}t_{[5,11]}}\,, \\
\nonumber A_{11}^{\phi^5:(2)} = & \, \frac{1}{t_{[6,9]}}\left(\frac{1}{t_{[2,5]}}+\frac{1}{t_{[5,11]}}\right)\,, \hspace{5mm} A_{11}^{\phi^5:(3)} = \, \frac{1}{t_{[7,10]}}\left(\frac{1}{t_{[2,5]}}+\frac{1}{t_{[5,11]}}\right)\,, \\
\nonumber A_{11}^{\phi^5:(4)} = & \, \frac{1}{t_{[3,6]}}\left(\frac{1}{t_{[8,11]}}+\frac{1}{t_{[2,8]}}\right)\,, \hspace{5mm} A_{11}^{\phi^5:(5)} = \, \frac{1}{t_{[4,7]}}\left(\frac{1}{t_{[8,11]}}+\frac{1}{t_{[2,8]}}\right)\,, \\
\nonumber A_{11}^{\phi^5:(6)} = & \, \frac{1}{t_{[3,9]}}\left(\frac{1}{t_{[3,6]}}+\frac{1}{t_{[6,9]}}\right)\,, \hspace{6mm} A_{11}^{\phi^5:(7)} = \, \frac{1}{t_{[7,10]}t_{[3,6]}}\,, \\
\nonumber A_{11}^{\phi^5:(8)} = & \, \frac{1}{t_{[4,10]}}\left(\frac{1}{t_{[4,7]}}+\frac{1}{t_{[7,10]}}\right)\,, \hspace{4mm} A_{11}^{\phi^5:(9)} = \, \frac{1}{t_{[4,7]}t_{[3,9]}}\,, \\
A_{11}^{\phi^5:(10)} = & \, \frac{1}{t_{[5,8]}t_{[3,9]}}\,, \hspace{5mm} A_{11}^{\phi^5:(11)} = \, \frac{1}{t_{[5,8]}t_{[4,10]}}\,, \hspace{5mm} A_{11}^{\phi^5:(12)} = \, \frac{1}{t_{[6,9]}t_{[4,10]}}\,.
\end{align}
The amplitude $A_{11}^{\phi^5}$ is the sum over all $\textrm{FC}_3(3,1)=12$ contributions and gives rise to the familiar expression with $\textrm{FC}_3(4,1)=22$ Feynman diagrams.

\end{example}

\vskip0.1in

From these examples note that even for $p>4$ the structure of the contribution of each region also resembles that of a cubic amplitude. In particular, it has the structure of $m_{(n+2(p-3))/(p-2)}(\alpha,\mathbb{I})$ for some permutation $\alpha$. Here we will only concentrate on the schematic structure of $A_n^{\phi^p}$ for all $n$, leaving again the precise connection between $\alpha$ and the region to future work.

\subsection{From Regions to Products of $\phi^3$ Amplitudes: Towards $m_{(n+2(p-3))/(p-2)}(\alpha,\mathbb{I})$}\label{genpSec}

As in the $p=4$ case, in order to understand the structure of each region it is useful to introduce an additional $(p-2)$-chord to the non-crossing $(p-2)$-chord diagrams from Definition \ref{kchords}. This is done by adding $p-2$ new points labelled $-p+3,-p+4,\dots,-1$ and $n-2$ so that the new set of points is $\{-p+3,-p+4,\ldots-1,0,1,2,\ldots,n-3,n-2\}$ and points are located in increasing order on the real line\footnote{In fact, any choice where at least one point is on the left of 0 and one point is on the right on $n-3$ is valid.}.

\begin{defn}
An extended non-crossing $(p-2)$-chord diagram is a non-crossing $(p-2)$-chord diagram on $n$ points labelled by $\{-p+3,-p+4,\ldots-1,0,1,2,\ldots,n-3,n-2\}$ in which $\theta_{-p+3,-p+4,...,-1,n-2}$ is always included. We also define a meadow of an extended non-crossing $(p-2)$-chord diagram as any region in the diagram delimited by more than one $(p-2)$-chord and by the line where the points lie.
\end{defn}

From now on we will abuse notation and use $\theta_{ab}$ to refer to the unique path in a $(p-2)$-chord joining two points $a$ and $b$. Therefore, the general claim is that a meadow delimited by $m$ such paths and the real line corresponds to a biadjoint ($m+1$)-subamplitude participating in $m_{(n+2(p-3))/(p-2)}(\alpha ,\mathbb{I})$. We also claim that the upper boundary of a meadow, $\theta_{ab}$, corresponds to a propagator in $m_{(n+2(p-3))/(p-2)}(\alpha ,\mathbb{I})$ of the form $1/t_{[a+3,b+2]}$, with the exception of the pair $\{a,b\}=\{-1,n-2\}$. This also fixes the topology of the cubic double-ordered amplitude.

Let us again give some examples to illustrate the proposal.

\begin{example}
Consider the extended non-crossing 4-chord diagram of $\phi^6$ for $n=14$ shown in figure \ref{n14p6}.

\begin{figure}[H]
\includegraphics[width=15.1cm]{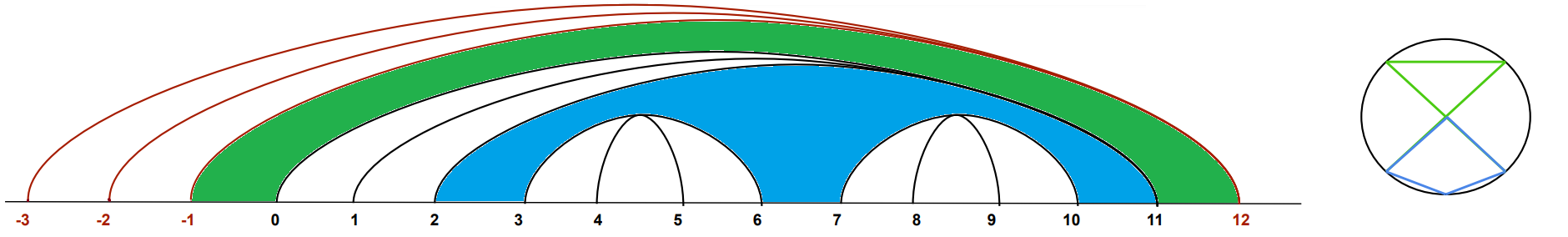}
\centering
\caption{\textit{Left:} An extended non-crossing chord diagram of $\phi^6$ for $n=14$ where the meadows have been coloured. \textit{Right:} Disc diagram of an $m_5(\alpha,\mathbb{I})$ amplitude corresponding to the contribution of the region on the left.}
\label{n14p6}
\end{figure}
In this extended diagram the points $-3$, $-2$, $-1$ and $12$ together with the 4-chord $\theta_{-3,-2,-1,12}$ that joins them are coloured in magenta. One can see that there are two meadows coloured in green and blue. The green meadow is delimited by the real line and by 2 paths $\theta_{-1,12}$ and $\theta_{0,11}$, thus it corresponds to a 3-point subamplitude appearing in $m_5(\alpha,\mathbb{I})$. Similarly, the blue meadow is delimited by 3 paths $\theta_{2,11}$, $\theta_{36}$ and $\theta_{7,10}$ and the real line, thus it corresponds to a 4-point subamplitude of $m_5(\alpha,\mathbb{I})$. The upper boundary of the blue meadow is $\theta_{2,11}$ and this means that there is a propagator of the form $1/t_{[5,13]}$. The upper boundary of the green meadow is of the form $\theta_{-1,n-2}$ and it does not generate a propagator. Therefore, the schematic form of the contribution is 

$$m_3\times m_4\times\frac{1}{t_{[5,13]}}\,.$$

\end{example}

\begin{example}
Now consider another extended non-crossing 4-chord diagram of $\phi^6$ for $n=14$ shown in figure \ref{n14p6v2}.

\begin{figure}[H]
\includegraphics[width=15.1cm]{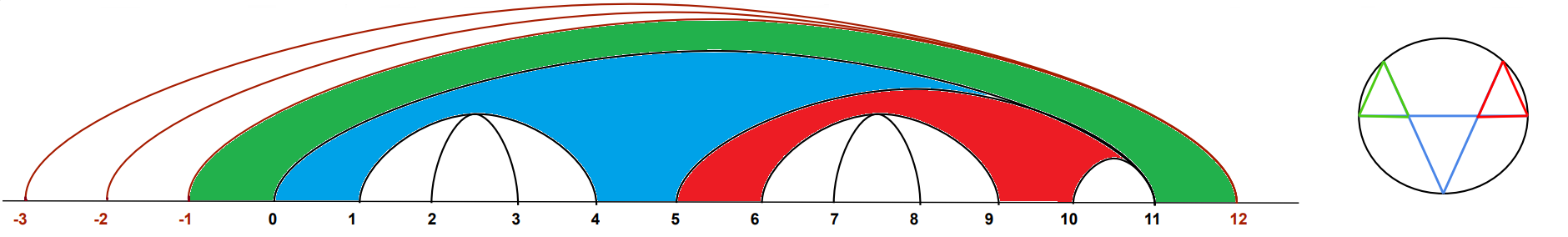}
\centering
\caption{\textit{Left:} An extended non-crossing chord diagram of $\phi^6$ for $n=14$ where the meadows have been coloured. \textit{Right:} Disc diagram of an $m_5(\alpha,\mathbb{I})$ amplitude corresponding to the contribution of the region on the left.}
\label{n14p6v2}
\end{figure}
As in the previous example, one can see that there are three meadows coloured in green, blue and red. The green meadow is delimited by the real line and by 2 paths $\theta_{-1,12}$ and $\theta_{0,11}$ and gives rise to a 3-point subamplitude appearing in $m_5(\alpha,\mathbb{I})$. Similarly, the blue meadow is delimited by the real line and 2 paths $\theta_{05}$ and $\theta_{14}$, thus it corresponds to a 3-point subamplitude of $m_5(\alpha,\mathbb{I})$. Likewise, the red meadow is delimited by the real line and 2 paths $\theta_{5,10}$ and $\theta_{69}$ and gives rise to a 3-point subamplitude of $m_5(\alpha,\mathbb{I})$. We also have two propagators of the form $1/t_{[3,7]}$ and $1/t_{[8,12]}$. Therefore, the schematic form of the contribution is 

$$(m_3)^3\times \frac{1}{t_{[3,7]}}\times \frac{1}{t_{[8,12]}}\,.$$

\end{example}

\subsection{Schematic Structure of $A_n^{\phi^p}$}

Before proposing a Lagrange inversion-like formula to reproduce the schematic structure of $A_n^{\phi^p}$, we present more examples:
\begin{align}\label{connT}
\nonumber A_8^{\phi^5} = & \,m_4+2m_3^2\frac{1}{P^2}\,, \quad
\nonumber A_{11}^{\phi^5}=  \,m_5+6m_3m_4\frac{1}{P^2}+5m_3^3\left(\frac{1}{P^2}\right)^2, \\
\nonumber A_{14}^{\phi^5}= & \,m_6+4m_4^2\frac{1}{P^2}+8m_5m_3\frac{1}{P^2}+28m_4m_3^2\left(\frac{1}{P^2}\right)^2+14m_3^4\left(\frac{1}{P^2}\right)^3, \\
\nonumber A_{10}^{\phi^6}= & \,m_4+3m_3^2\frac{1}{P^2}\,, \quad
\nonumber A_{14}^{\phi^6}=  \,m_5+9m_3m_4\frac{1}{P^2}+12m_3^3\left(\frac{1}{P^2}\right)^2, \\
\nonumber A_{12}^{\phi^7}= & \,m_4+4m_3^2\frac{1}{P^2}\,, \\
 A_{14}^{\phi^8}= & \,m_4+5m_3^2\frac{1}{P^2}\,.
\end{align}

Let us make a proposal for the all $n$ structure of $A^{\phi^p}_n$ amplitudes in terms of biadjoint cubic amplitudes and then perform the same consistency check as done for $\phi^4$.  

The proposal is motivated by the fact, proven in appendix \ref{lagrange}, that $f(x)=x/B_{k-1}(x)$ and $g(x)=x B_k(x)$ are compositional inverses of each other if $B_r(x)$ is the generating function of the Fuss-Catalan numbers $\textrm{FC}_m(r,1)$. This led us to propose a recursive structure in which we define
\be 
h_3(x):=\sum_{i=0}^\infty m_{i+2} x^i\,
\ee 
and 
\be\label{hRec} 
h_k(x)= \sum_{j=0}^\infty h_{k,j}x^j := \sum_{j=0}^\infty  \frac{1}{2\pi i}\oint_{|z|=\epsilon} \frac{dz}{j+1}\left(\frac{h_{k-1}(z)}{z}\right)^{j+1} x^j\,.
\ee 
The structure of the $A^{\phi^p}_n$ amplitude is then given by 
\be 
A^{\phi^p}_n = \frac{h_{p,(n-2)/(p-2)}}{h_0^{(p-3)(n-2)/(p-2)}}\,.
\ee 

The first consistency check is that the number of non-crossing $(p-2)$-chord diagrams with $(n-2)/(p-2)$ chords is $\textrm{FC}_{(n-2)/(p-2)}(p-2,1)$. Therefore if one sets all $m_{i+2}:=1$ so that the contribution from each region is unity, one must find that $A^{\phi^p}_n =\textrm{FC}_{(n-2)/(p-2)}(p-2,1)$.

In section \ref{conMi} we showed that setting
\be 
h_3(x) = \sum_{i=1}^\infty x^i = \frac{1}{1-x}=B_1(x)\,,
\ee 
turns $h_4(x)$ into the generating function of the numbers $\textrm{FC}_m(3,1)$. Iterating the procedure one finds that $h_k(x)$ turns into the generating function of the numbers $\textrm{FC}_m(k-1,1)$. 

The second check is evaluating $A^{\phi^p}_n$ on planar kinematics so that $A^{\phi^p}_n$ counts the number of unrooted planar $(p-1)$-ary trees (with all internal vertices of degree $p$). The numbers are known to be given by the Fuss-Catalan sequence, ${\rm FC}_{(n-2)/(p-2)}(p-1,1)$. This check can again be done by realizing that on planar kinematics $m_n=\textrm{C}_{n-2}$ and therefore
\be 
h_3(x) = \sum_{i=1}^\infty \textrm{C}_i\, x^i =B_2(x) =\frac{1-\sqrt{1-4x}}{2x}\,.
\ee 
Iterating one finds that $h_k(x) = B_{k-1}(x)$, the generating function of the Fuss-Catalan numbers $\textrm{FC}_m(k-1,1)$ as required. 

In appendix \ref{lagrange} we provide several examples that illustrate the iteration procedure and the resulting formulas for $A_n^{\phi^p}$.

%(see OEIS entry A045623,  \cite{oeis})

\section{Discussions}\label{discs}

In this work we have extended the global Schwinger formulation to all partial amplitudes $m_n(\alpha,\beta)$, and also to amplitudes in $\phi^p$ theories. $A_n^{\phi^p}$ is given as a sum over regions, each of which is proposed to be in bijection with a $\phi^3$ biadjoint partial amplitude. This leads to the statement that $A_n^{\phi^p}$ amplitudes can be understood as a sum of products of cubic amplitudes. 

A very simple diagrammatic procedure for listing all regions contributing to an amplitude was found in terms of non-crossing k-chord diagrams. Given one such diagram, we have provided an algorithm for determining the structure of the contribution in terms of $\phi^3$ amplitudes (meadows) and propagators (``frontiers'' separating meadows). Every meadow can be seen to be related to a cubic amplitude participating in $m_{(n+2(p-3))/(p-2)}(\alpha,\mathbb{I})$. Our identification so far is lacking a direct way of determining the permutation $\alpha$ from the non-chord diagram. It would also be very important to find a purely combinatorial method to determine the precise bijection between the set of planar kinematic invariants in each object.

Our main focus has been on a combinatorial prescription for $\phi^p$ amplitudes. However, it would be interesting to find a diagrammatic procedure, in the lines of that for $\phi^p$ amplitudes, to determine the regions that compute $m_n(\alpha,\beta)$ so that $H_{\alpha}(x)=0$ in \eqref{aux1}. 

We end this work with three topics for future research.

\subsection{Relation to Green Functions in Planar Theories}

The standard way of computing Green functions, $G_n$, from connected Green functions, $G_n^c$, is via an exponentiation procedure. However, it is well-known that in planar theories this does not work \cite{Brezin:1977sv}. This is because planarity forces points of the Green function $G_n(x_1,x_2,\ldots ,x_n)$, to be on the boundary of a disk and a  connected Green function for points in a subset $J\subset \{ x_1,x_2,\ldots ,x_n\}$ can be thought of as cutting the disk into regions so that one of them only contains the points in $J$. Having done this, another connected Green function can only be constructed from the pieces left, and so on. In \cite{Brezin:1977sv}, a theory that only admits Green functions with an even number $n=2q$ of points was considered and gave rise to the following combinatorial problem: in how many ways can $2q$ points on a disk be clustered in non-overlapping sets so that there are $r_1$ pairs, $r_2$ quadruplets, etc. As shown in \cite{Brezin:1977sv} this is solved by the coefficients of the formula relating Green functions 
\be\label{greenFS} 
G_{2q} = \sum_{r_i\geq 0}\delta_{q,\left(\sum_{i}i\, r_i\right)}\frac{(2q)!}{(2q+1-\sum_i r_i)!} \frac{\left(G_2^c\right)^{r_1}}{r_1!}\frac{\left(G_4^c\right)^{r_2}}{r_2!}\ldots \frac{\left(G_{2q}^c\right)^{r_{q}}}{r_{q}!}\,.
\ee 
Here the Kronecker delta guarantees that each of the $2q$ points participates in each term. 

Up to this point, this discussion seems to be completely independent of the formulas found in this work. In order to see the connection, let us list the first few cases as done in eq. (31) of \cite{Brezin:1977sv},
\begin{align*}
    G_4 = &\, G_4^c + 2\left(G_2^c\right)^2, \\
    G_6 = &\, G_6^c + 6\, G_4^c G_2^c+ 5\left(G_2^c\right)^3, \\
    G_8 = &\, G_8^c + 4\left(G_4^c\right)^2+ 8\, G_6^cG_2^c+ 28\, G_4^c\left(G_2^c\right)^2+14\left(G_2^c\right)^4\,.
\end{align*} 
Comparing to the expressions for $\phi^5$ amplitudes in \eqref{connT}, i.e. 
\begin{align*}
A_8^{\phi^5} = & \,m_4+2m_3^2\frac{1}{P^2}\,, \\
A_{11}^{\phi^5}= & \,m_5+6m_3m_4\frac{1}{P^2}+5m_3^3\left(\frac{1}{P^2}\right)^2\,, \\
A_{14}^{\phi^5}= & \,m_6+4m_4^2\frac{1}{P^2}+8m_5m_3\frac{1}{P^2}+28m_4m_3^2\left(\frac{1}{P^2}\right)^2+14m_3^4\left(\frac{1}{P^2}\right)^3\,,
\end{align*}
it is clear that there must be a relation. The fact that the coincidence of the structure continues to all multiplicities is shown using the Lagrange inversion formula in appendix \ref{lagrange}.

It is natural to expect that the relation  extends to all $\phi^p$ amplitudes as follows. Let $m=p-3$, and place $mq$ points on a disk. Now count all possible ways of clustering the points in non-overlapping sets so that there are $r_1$ groups of $m$ points each, $r_2$ groups of $2m$ points, etc. Then the formula that relates $A^{\phi^p}_n$ amplitudes and Green functions is given by the natural generalization of \eqref{greenFS},
\be\label{newGreen} 
G_{mq} = \sum_{r_i\geq 0}\delta_{q,\left(\sum_{i}i\, r_i\right)}\frac{(mq)!}{(mq+1-\sum_i r_i)!} \frac{\left(G_m^c\right)^{r_1}}{r_1!}\frac{\left(G_{2m}^c\right)^{r_2}}{r_2!}\ldots \frac{\left(G_{mq}^c\right)^{r_{q}}}{r_{q}!}\,.
\ee 
We leave it as an exercise for the reader to check that the coefficients we have presented in the text and the ones in appendix \ref{lagrange} are indeed the correct values of the combinatorial problem and the coefficients in \eqref{newGreen}.

It would be very interesting to explore this connection further, in particular to matrix models with $\Phi^{p-1}$ interactions as the one studied in \cite{Brezin:1977sv}.

\subsection{Possible Connection with Stokes Polytopes}

Recent work on the computation of $\phi^p$ amplitudes as a sum over contributions obtained from various polytopes known as accordiohedra is very reminiscent of the structures we have uncovered using the global Schwinger formulation. Developing a connection between the two approaches is certainly an important problem. Here we restrict to $\phi^4$ amplitudes and therefore to Stokes polytopes in order to point out some possible directions. Most of the formulations using Stokes polytopes construct the amplitudes as (see e.g equation (5) of \cite{Raman:2019utu})
\be\label{origS}  
A^{\phi^4}_n = \sum_{{\rm Symmetry:}\, \sigma}\,\,\sum_{{\rm Primitive:}\, P} \alpha_P\, m_{P,n}^{(\sigma.P)}
\ee 
where the sum is over all primitive Stokes polytopes and the symmetry classes into which they fall. The $m_{P,n}^{(\sigma.P)}$ are the contributions obtained from the corresponding polytope. Here the $\alpha_P$ are the so-called {\it weights}, which are in general rational numbers. 

Consider for example, 
\be 
A^{\phi^4}_6 = \alpha_1\left(\frac{1}{X_{1,4}}+\frac{1}{X_{3,6}} \right)+ \alpha_2\left(\frac{1}{X_{2,5}}+\frac{1}{X_{1,4}} \right)+\alpha_3\left(\frac{1}{X_{3,6}}+\frac{1}{X_{2,5}} \right)\,.
\ee 
Here there are three polytopes and the weights have to be chosen to be $\alpha_a = 1/2$.

In \cite{Salvatori:2019phs}, Salvatori and Stanojevic propose a way to simplify \eqref{origS} by reducing the redundancy by taking certain limits of kinematic invariants in each term. Let us rewrite Eq. 4.8 of \cite{Salvatori:2019phs} for $n=6$,
\be\label{ss6} 
A^{\phi^4}_6 = \left(\frac{1}{X_{1,4}}+\frac{1}{X_{3,6}} \right)+\lim_{X_{1,4}\to \infty} \left(\frac{1}{X_{2,5}}+\frac{1}{X_{1,4}} \right)\,.
\ee 
In this formula, the first bracket comes from the Stokes polytope with reference $1,4$ while the second bracket comes from the reference $2,5$. Here $X_{i,j}$ can be identified with the planar invariants $t_{[a,b]}$ in a simple way. Note that \eqref{ss6} groups the three terms in the same way as that found in our construction \eqref{our1} and \eqref{our2} coming from the two possible non-crossing chord diagrams. In \cite{Salvatori:2019phs}, the $n=8$ amplitude is also computed. The amplitude is given as a sum over five Stokes polytopes. Our formula \eqref{qmet} also has five regions. However, while our regions all contribute with a factor of one, Eq. 4.10 of \cite{Salvatori:2019phs} has four terms with coefficient $+1$ and one with $-1$. In fact, only the first region can be matched directly; it coincides with the first polytope, i.e. the one with no limits and which gives rise, in our language, to $m_5(\mathbb{I},\mathbb{I})$. We suspect that there exist other combinations with different limits which could match our formula term by term. One hint is that every one of our terms is isomorphic to either an associahedron or to intersections of two associahedra.

\subsection{Towards Generalized $\phi^p$ Amplitudes}

Another intriguing feature of our procedure for constructing the regions from the extended non-crossing chord diagrams for $\phi^4$ is the introduction of two additional points ($-1$ and $n-2$) and a chord joining them. The relevance of this additional chord lies in the way each meadow is associated to a cubic amplitude. For now we have conceived these diagrams simply as combinatorial objects, but if one attempts to relate each of the labels $0,1,\ldots,n-3$ in the diagram to the particles $3,4,\ldots,n$ respectively, then the new chord $\theta_{-1,n-2}$ has the interpretation of identifying particles 1 and 2, which in their tropicalized variables are set to $-\infty$ and $+\infty$. However, one has to be careful in that the variable $x_{a-3}$ coming from the parameterization \eqref{preT} does not exactly correspond to a single particle $a$ as it appears in all rows $r\geq a$. It might seem puzzling that for a general value of $p$, we introduce $p-2$ additional points to the non-crossing $(p-2)$-chord diagrams and join them with another $(p-2)$-chord. Strikingly, $p-2$ is precisely the number of particles with tropicalized variables set at infinity that appear in the higher-$k$ version of the global Schwinger parameterization using $\textrm{Trop}^+G(k,n)$ for $k=p-2$ (see \cite{Cachazo:2020wgu} for its construction). It would be interesting to explore if there is a connection with these generalized objects and CEGM generalized amplitudes \cite{Cachazo:2019ngv}.

In fact, one direction to tackle is to try and find an analog of $\phi^p$ amplitudes for higher-$k$ theories, using a similar limiting procedure. One possible direction is the following. In \cite{Borges:2019csl,Cachazo:2019xjx} it was shown that \textit{planar arrays} of Feynman diagrams computed higher-$k$ amplitudes, also known as CEGM amplitudes \cite{Cachazo:2019ngv}. An element in the planar array ${\cal A}_{i_1,i_2\ldots ,i_{k-2}}$ is an $(n-k+2)$-particle cubic Feynman diagram with particle labels $\{ 1,2,\ldots ,n\}\setminus \{i_1,i_2\ldots ,i_{k-2}\}$. By introducing the notation ${\cal I}:= \{i_1,i_2\ldots ,i_{k-2}\}$ one can parameterize the higher-$k$ invariants as \cite{Cachazo:2019xjx,Early:2020hap}
\be
\sfs_{i_1i_2\ldots i_k}:= \sum_{{\cal I}\cup \{j_1,j_2\}=\{i_1,i_2,\ldots ,i_k\}} s_{j_1,j_2}^{({\cal I})}\,,
\ee
where the sum is over all possible ways of decomposing $\{i_1,i_2,\ldots ,i_k\}$ into two sets of $k-2$ and $2$ elements respectively. For example, for $k=3$ we have
\be\label{clu}
\sfs_{ijk} := s^{(i)}_{jk} + s^{(j)}_{ki} + s^{(k)}_{ij}\,.
\ee
This means that we can write each kinematic invariant appearing in an element in the array in terms of the planar invariants of that particular Feynman diagram, and then apply the familiar limiting procedure to it. After this we obtain the rank-$k$ kinematic invariants which are sent to $1/\epsilon$.

As an example, consider $m^{(3)}_7(\mathbb{I},\mathbb{I})$. Each element in an array, or when $k=3$ also known as a collection, is a 6-particle binary tree. The proposed limiting procedure to get a $p=4$ object imposes that kinematic invariants of the form $s^{(a)}_{bc}$ are sent to $1/\epsilon$ if $\{ b,c\}$ are consecutive in the ordering $(1,2,\ldots ,\hat a,\ldots, n)$. 

In this case we have that, e.g., $\{s^{(1)}_{23}, s^{(1)}_{34},\ldots , s^{(1)}_{67}, s^{(1)}_{72}, s^{(2)}_{13}, s^{(2)}_{34}, \ldots ,s^{(2)}_{67}, s^{(2)}_{71}\}$ are sent to $1/\epsilon$. Notice that after the limiting procedure each of the 693 collections will either end up being of order ${\cal O}(\epsilon^4)$ or higher. Therefore one can define $A^{\phi^4}_{3,7}$ as  the coefficient of $\epsilon^4$ in a series expansion around $\epsilon\to 0$.

Rather than starting with the CEGM generalized amplitude, a better understood object is the set of PK amplitudes introduced by Early in \cite{Early:2021tce}. We leave the study of this construction to future research.

\section*{Acknowledgements}

The authors thank N. Early for helpful discussions and comments on the draft, and for identifying the regions that define the support of the distribution $Q(x)$ with polyhedral cones. We also thank N. Kalyanapuram for bringing some references about Stokes polytopes and accordiohedra to our attention. FC thanks S. Mizera, G. Zhang, and K. Yeats for discussions related to constructions of CHY formulations of $\phi^p$ amplitudes. This research was supported in part by a grant from the Gluskin Sheff/Onex Freeman Dyson Chair in Theoretical Physics and by Perimeter Institute. Research at Perimeter Institute is supported in part by the Government of Canada through the Department of Innovation, Science and Economic Development Canada and by the Province of Ontario through the Ministry of Colleges and Universities.

\appendix

\section{Lagrange Inversion Formula and Fuss-Catalan Numbers}\label{lagrange}

Given a function $f(x)$ that admits a series expansion around $x=0$ and $f(0)=0$ while $f'(0)\neq 0$, the Lagrange inversion formula gives a series expansion for the compositional inverse of $f(x)$, i.e. for a function $g(x)$ such that $g(f(x))=x$, in terms of the series coefficients of $f(x)$. 

Let us review one particular formulation which is relevant for this work. Start by defining an auxiliary function $h(x)$ such that $h(x)=x/f(x)$. Let 
\be 
h(x) =\sum_{i=0}^\infty h_i x^i
\ee 
be the series expansion of $h(x)$ around $x=0$. Now define
\be\label{expan} 
g_r : = \frac{1}{2\pi i}\oint_{|z|=\epsilon} \frac{dz}{r+1}\left(\frac{h(z)}{z}\right)^{r+1} .
\ee 
The Lagrange inversion formula states that the series expansion of $g(x)$ is of the form
\be\label{gim} 
g(x) = x\sum_{r=0}^\infty g_r x^r\,.
\ee 
The proof is fairly simple. Consider the RHS of \eqref{expan} and write it in terms of $f(z)$ and then write $z=g(u)$,
\be\label{fime} 
\frac{1}{2\pi i}\oint_{|z|=\epsilon} \frac{dz}{r+1}\left(\frac{1}{f(z)}\right)^{r+1}\!\!\! = \frac{1}{2\pi i}\oint_{|u|=\epsilon} \frac{du}{r+1}g'(u)\left(\frac{1}{f(g(u))}\right)^{r+1}\!\!\! = \frac{1}{2\pi i}\oint_{|u|=\epsilon} \frac{du}{(r+1)}\frac{g'(u)}{u^{r+1}}.
\ee
Now, taking the derivative of \eqref{gim}, 
\be 
g'(x) =\sum_{r=0}^\infty (r+1)g_rx^r
\ee 
and plugging in it into the last expression on the right in \eqref{fime} one finds \eqref{expan}.

Before seeing explicitly how this works in the context of interest, let us review some well-known facts about generating functions of Fuss-Catalan numbers, in particular, how they are interconnected via the Lagrange inversion formula.  

Let $B_k(x)$ be the generating of the Fuss-Catalan numbers ${\rm FC}_n(k,1)$. The function $B_k(x)$ satisfies the equation $B_k(x) = 1+x B_k(x)^k$. Now let us prove that
\be 
f(x):=\frac{x}{B_{k-1}(x)}, \quad g(x):= x B_{k}(x)
\ee 
are compositional inverses of each other. Start with $B_k(x) = 1+x B_k(x)^k$ and multiply by $x^{k-1}$ so that we get an equation for $g(x)$ of the form $x^{k-2}g(x) = x^{k-1}+g(x)^k$. Now let $p(x)$ be the compositional inverse of $g(x)$, i.e., $g(p(x))=x$. Letting $x=p(u)$ in $x^{k-2}g(x) = x^{k-1}+g(x)$ leads to $u\, p(u)^{k-2} = p(u)^{k-1}+u^k$. Let us prove that $f(x)$ satisfies the same equation as $p(x)$. Starting with $B_{k-1}(x) = 1+x B_{k-1}(x)^{k-1}$ and substituting $B_{k-1}(x)=x/f(x)$ gives $x/f(x) = 1+x^k/f(x)^{k-1}$. Multiplying by $f(x)^{k-1}$ we obtain the same equation satisfied by $p(x)$. 

In the case at hand, we are interested in $k=3$ so that 
\be 
h(x) = B_2(x) =\frac{1-\sqrt{1-4z}}{2z}
\ee 
is the generating function of Catalan numbers, and $g_r$ becomes the Fuss-Catalan number ${\rm FC}_r(3,1)$. In this case \eqref{expan} reads
\be\label{threeProof}
{\rm FC}_r(3,1) = \frac{1}{2\pi i}\oint_{|z|=\epsilon} \frac{dz}{r+1}\left(\frac{1-\sqrt{1-4z}}{2z^2}\right)^{r+1}
\ee 
with
\be 
g(x) = x B_3(x), \quad {\rm with}\quad B_3(x) := \sum_{r=0}^\infty {\rm FC}_r(3,1) x^r\,.
\ee 

Let us see how this applies to our construction in section \ref{combi}. Let us consider the following choice for the function $h(x)$ 
\be 
h(x) =\sum_{i=0}^\infty m_{i+2} x^i
\ee 
where $m_{i+2}$ represents a generic $(i+2)$-particle amplitude in the biadjoint $\phi^3$ scalar theory of the form $m_{i+2}(\mathbb{I},\mathbb{I})$. 

Since the mass dimension of $m_{i+2}(\mathbb{I},\mathbb{I})$ is $-2(i-1)$ we are motivated to define $m_2:=P^2$ and $m_3:=1$. Here $1/P^2$ represents a generic propagator. We will soon see why this somewhat strange definition of $m_2$ is useful. Let us start by noticing that the number of Feynman diagrams contributing to $m_{i+2}(\mathbb{I},\mathbb{I})$ is the Catalan number $\textrm{C}_{i}$.

The claim is that the form of the amplitude $A^{\phi^4}_n$ is determined by the coefficient $g_{n/2-1}$ divided by $h_0^{n/2-1}$. 

Let us compute the first few cases of \eqref{expan} in order to illustrate the use of the formula,
\begin{align*}
A^{\phi^4}_4 = \frac{g_{1}}{h_0} = & h_1  = m_3 \\
A^{\phi^4}_6 = \frac{g_{2}}{h_0^2} = & \frac{h_0^2h_2+h_0h_1^2}{h_0^2} = m_4+m_3^2\frac{1}{P^2} \\
A^{\phi^4}_8 = \frac{g_{3}}{h_0^3} = & \frac{h_3 h_0^3+3 h_1 h_2 h_0^2+h_1^3 h_0}{h_0^3} = m_5+3m_3m_4\frac{1}{P^2}+m_3^3\left(\frac{1}{P^2}\right)^2.
\end{align*}

Finally, specializing to what is called {\it planar kinematics}, in which all planar invariants are set to unity, one finds that $A^{\phi^4}_n$ counts the number of Feynman diagrams contributing to the amplitude. This is the number of ternary planar unrooted trees with $n$ leaves which is known to be the given by the Fuss-Catalan numbers. Applying the same kinematics to the $\phi^3$ amplitudes one can replace each by the corresponding Catalan numbers and therefore we reproduce the relation \eqref{expan}. 

\subsection{Extension to $\phi^p$: Iterated Structure}

Let us explicitly construct the iteration used in section \ref{genpSec} to propose the schematic structure of $\phi^p$ amplitudes. 

Let us start by defining generating functions
\be 
h_k(x)=\sum_{j=0}^\infty h_{k,j}x^j\,.
\ee 
The goal is to construct a recursive procedure that determines all coefficients $h_{k,j}$ as functions of the base case defined to be
\be 
h_3(x)=\sum_{j=0}^\infty h_{j} x^j\,.
\ee 
Note that for the base function we have denoted the coefficients by $h_j$ instead of $h_{3,j}$. This was done in order not clutter the formulas. 
Using the expression in \eqref{hRec}
\be 
h_k(x)= \sum_{j=0}^\infty h_{k,j}x^j := \sum_{j=0}^\infty  \frac{1}{2\pi i}\oint_{|z|=\epsilon} \frac{dz}{j+1}\left(\frac{h_{k-1}(z)}{z}\right)^{j+1} x^j
\ee 
let us present some results for the expansions.

For $\phi^4$ amplitudes we have $h_4(x)$ with coefficients
\be\label{phi4Num} 
\begin{array}{l}
 h_0 \\
 h_0 h_1 \\
h_0 h_1^2+h_0^2 h_2 \\
 h_0 h_1^3+3 h_0^2 h_1 h_2+h_0^3 h_3 \\
 h_0 h_1^4+6 h_0^2 h_1^2 h_2+2 h_0^3 h_2^2+4 h_0^3
   h_1 h_3+h_0^4 h_4 \\
h_0 h_1^5+10 h_0^2 h_1^3 h_2+10 h_0^3 h_1 h_2^2+10
   h_0^3 h_1^2 h_3+5 h_0^4 h_2 h_3+5 h_0^4 h_1 h_4+h_0^5
   h_5 .\\
\end{array}
\ee 

These coefficients are a refinement of the Narayana numbers. Let us see this more explicitly. Consider first the table of coefficients (see OEIS entry A134264,  \cite{oeis}),
\be 
\begin{array}{l}
 1 \\
 1 \\
 1,1 \\
 1,3,1 \\
 1,6,2,4,1 \\
 1,10,10,10,5,5,1 \\
 1,15,30,5,20,30,3,15,6,6,1 \\
 1,21,70,35,35,105,21,21,35,42,7,21,7,7,1 .\\
\end{array}
\ee
If we now set $h_0=x$ and all other $h_i=1$, then terms with the same power of $h_0$ are combined. For example, $2 h_2^2 h_0^3+4 h_1 h_3 h_0^3$ in the fifth row of \eqref{phi4Num} becomes $2x^3+4x^3=6x^3$. Carrying this out one gets
\be 
\begin{array}{l}
 x \\
 x \\
 x^2+x \\
 x^3+3 x^2+x \\
 x^4+6 x^3+6 x^2+x \\
 x^5+10 x^4+20 x^3+10 x^2+x \\
 x^6+15 x^5+50 x^4+50 x^3+15 x^2+x \\
 x^7+21 x^6+105 x^5+175 x^4+105 x^3+21 x^2+x .\\
\end{array}
\ee
These coefficients are the Narayana numbers (OEIS entry A001263,  \cite{oeis}). 

Let us consider $\phi^5$ amplitudes, so we have $h_5(x)$ with coefficients
\be\label{phi5Num} 
\begin{array}{l}
 h_0 \\
h_0^2 h_1\\
2 h_0^3 h_1^2+h_0^4 h_2 \\
5 h_0^4 h_1^3+6 h_0^5 h_1 h_2+h_0^6 h_3 \\
14 h_0^5 h_1^4+28 h_0^6 h_1^2 h_2+4 h_0^7 h_2^2+8
   h_0^7 h_1 h_3+h_0^8 h_4 \\
42 h_0^6 h_1^5+120 h_0^7 h_1^3 h_2+45 h_0^8 h_1
   h_2^2+45 h_0^8 h_1^2 h_3+10 h_0^9 h_2 h_3+10 h_0^9 h_1
   h_4+h_0^{10} h_5. \\
\end{array}
\ee 
Listing only the coefficient allows us to present one more row (OEIS entry A338135,  \cite{oeis}),
\be 
\begin{array}{l}
 1 \\
 1 \\
 2,1 \\
 5,6,1 \\
 14,28,4,8,1 \\
 42,120,45,45,10,10,1 \\
 132,495,330,22,220,132,6,66,12,12,1 . \\
\end{array}
\ee 
Once again, if we set $h_0=x$ and all other $h_i=1$, then \eqref{phi5Num} becomes the generating functions for the $2$-Narayana numbers. In general one finds the triangle of $m$-Narayana numbers, where the standard ones correspond to $m=1$. The $2$-Narayana numbers are then given by the coefficients in (see e.g. section 6.8 of \cite{nara})
\be 
\begin{array}{l}
 x \\
 x^2 \\
 x^4+2 x^3 \\
 x^6+6 x^5+5 x^4 \\
 x^8+12 x^7+28 x^6+14 x^5 \\
 x^{10}+20 x^9+90 x^8+120 x^7+42 x^6 \\
 x^{12}+30 x^{11}+220 x^{10}+550 x^9+495 x^8+132 x^7 \\
% x^{14}+42 x^{13}+455 x^{12}+1820 x^{11}+3003 x^{10}+2002
%   x^9+429 x^8. \\
\end{array}
\ee 

\subsection{One Function to Compute Them All}

There is one more interesting property of these representation of $\phi^p$ amplitudes which interconnects them. Consider the coefficients of the function $h_4(x)$. Some of them are explicitly shown in \eqref{phi4Num}. 

The claim is that the coefficients of the function $h_k(x)$ can be obtained from those of $h_4(x)$ by simply setting to zero all $h_a$ with $a \notin (k-3)\mathbb{Z}$ (see text in OEIS entry A338135 for $k=5$ case, \cite{oeis}). For example, $h_5(x)$ is obtained by setting all $h_a$ with $a$ odd to zero. Of course, every other coefficient of $h_4(x)$ vanishes completely but the ones that do not reproduce $h_5(x)$.

One direct way to understand the relation among the different generating functions $h_k(x)$ is by recalling the combinatorial problem they solve. As explained in the discussions, set $m=k-3$ and place $mq$ points on a disk. Now count all possible ways of clustering the points in non-overlapping sets so that there are $r_1$ groups of $m$ points each, $r_2$ groups of $2m$ points, etc. Clearly, $h_4(x)$, for which $m=1$, contains all other problems counted by $h_k(x)$ with $k>4$ as special cases.

\section{Computing a Region for $n=12$ that Leads to $m_7(1234576,\mathbb{I})$}\label{appB}

Directly computing amplitudes, $A^{\phi^4}_n$, using the global Schwinger formula becomes harder as $n$ grows. In this appendix, we show how to use the global Schwinger formula to find an explicit map from a region to $m_{n/2+1}(\alpha ,\mathbb{I})$. Having the precise bijection of kinematic invariants, one gets the contribution of the region without ever carrying out an integral. 
The region under consideration is 
\be 
R=\{ x_0 = x_3 <  x_1 = x_2,\, x_4=x_5,\, x_6=x_7,\, x_8=x_9 \} \,.
\ee 

Let us consider the behavior of the $G_{12}(x)$ part of the tropical potential on a region $R_{\rm ext}$ where the condition $x_3<x_1$ is relaxed an therefore contains $R$, i.e.,
\be
R \subset R_{\rm ext} = \{ x_0 = x_3 \, x_1 = x_2,\, x_4=x_5,\, x_6=x_7,\, x_8=x_9 \} \,.
\ee 
The function $G_{12}(x)$ can be seen to be a linear combination of the following $14$ piecewise linear functions
\begin{align*}
    & \min \left(x_1,x_3\right), \, \min \left(x_3,x_5\right),\, \min
   \left(x_5,x_7\right),\, \min \left(x_7,x_9\right), \\ 
   & \min
   \left(x_1,x_3,x_5\right),\, \min \left(x_3,x_5,x_7\right),\, \min
   \left(x_5,x_7,x_9\right),\,  \\ &  \min \left(x_1,x_3,x_5,x_7\right), \min
   \left(x_3,x_5,x_7,x_9\right),\min
   \left(x_1,x_3,x_5,x_7,x_9\right), \\ & x_1,x_5,x_7,x_9.
\end{align*}
Note that $x_1$ is always accompanied by $x_3$ when it is an argument in a $\min$ function. This means that when restricting to $R$, i.e. imposing $x_3<x_1$ on the functions, $x_1$ drops out and we are left with the following $11$ functions,
\begin{align*}
    & \min \left(x_3,x_5\right),\, \min
   \left(x_5,x_7\right),\, \min \left(x_7,x_9\right), \\ 
   & \min \left(x_3,x_5,x_7\right),\, \min
   \left(x_5,x_7,x_9\right),\,  \\ &  \min
   \left(x_3,x_5,x_7,x_9\right), \, x_1,\, x_3,\, x_5,\, x_7,\, x_9.
\end{align*}
It is easy to compute the coefficients of each of the $11$ functions to be
\begin{align*}
   & t_{[2,8]}-t_{[2,6]}-t_{[6,8]},\,-t_{[6,8]}+t_{[6,10]}-t_{[8,10]},\, -t_{[8,10]}+t_{[8,12]}-t_{[10
   ,12]},\\ & t_{[2,10]}-t_{[2,8]}+t_{[6,8]}-t_{[6,10]},\, -t_{[6,10]}+t_{[6,12]}+t_{[8,10]}-t_{[8,12]}, \\ & t_{[6,10]}-t_
   {[2,10]}-t_{[6,12]}, \, t_{[3,5]},\, t_{[2,6]}-t_{[3,5]},\,t_{[6,8]},\,t_{[8,10]},\,t_{[10,12]}.
\end{align*}
Note that the coefficient of $x_1$ is $t_{[3,5]}$, which is precisely the invariant in the propagator that must appear according to the rules for the non-crossing diagram corresponding to the region $R$. The only other place where $t_{[3,5]}$ appears is in the coefficient of $x_3$. This means that we can write the integral over $x_1$ as
\be 
\int_{-\infty}^\infty dx_1 \theta(x_1-x_3)\exp \left( - t_{[3,5]}(x_1-x_3)\right)  = \frac{1}{t_{[3,5]}}\,.
\ee 

Combining the left over terms and relabeling variables so that $x_a\to x_{(a-3)/2}$ one finds the ``effective'' potential 
\begin{align*}
F_6 & =  (t_{[2,8]}-t_{[2,6]}-t_{[6,8]})\min \left(x_0,x_1\right)+( -t_{[6,8]}+t_{[6,10]}-t_{[8,10]}) \min
   \left(x_1,x_2\right) \\ & +(-t_{[8,10]}+t_{[8,12]}-t_{[10
   ,12]}) \min \left(x_2,x_3\right) +(t_{[2,10]}-t_{[2,8]}+t_{[6,8]}-t_{[6,10]})\min \left(x_0,x_1,x_2\right) \\ & +( -t_{[6,10]}+t_{[6,12]}+t_{[8,10]}-t_{[8,12]})\min
   \left(x_1,x_2,x_3\right) \\ & + (t_{[6,10]}-t_
   {[2,10]}-t_{[6,12]})\min
   \left(x_0,x_1,x_2,x_3\right)+  t_{[2,6]} x_0+t_{[6,8]}x_1+t_{[8,10]}x_2+t_{[10,12]}x_3\,.
\end{align*}  
It is a simple exercise to match the coefficients with that of the tropical potential function for $m_6(\mathbb{I},\mathbb{I})$. The non-trivial fact is that the result is not only a map but a bijection between the corresponding sets of planar invariants. This is left as an exercise for the reader.

\bibliographystyle{JHEP}
\bibliography{references}

\providecommand{\href}[2]{#2}\begingroup\raggedright\begin{thebibliography}{10}

\bibitem{SWTrop}
D.~{Speyer} and L.~K. {Williams}, {\it {The tropical totally positive
  Grassmannian}},  {\em arXiv Mathematics e-prints} (Dec., 2003) math/0312297,
  [\href{http://arxiv.org/abs/math/0312297}{{\tt math/0312297}}].

\bibitem{Cachazo:2020wgu}
F.~Cachazo and N.~Early, {\it {Planar Kinematics: Cyclic Fixed Points, Mirror
  Superpotential, k-Dimensional Catalan Numbers, and Root Polytopes}},
  \href{http://arxiv.org/abs/2010.09708}{{\tt arXiv:2010.09708}}.

\bibitem{yuliy}
Y.~Baryshnikov, {\it On stokes sets},  {\em New developments in singularity
  theory (Cambridge, 2000)} {\bf 21} (2001) 65--86.

\bibitem{Banerjee:2018tun}
P.~Banerjee, A.~Laddha, and P.~Raman, {\it {Stokes polytopes: the positive
  geometry for $\phi^{4}$ interactions}},  {\em JHEP} {\bf 08} (2019) 067,
  [\href{http://arxiv.org/abs/1811.05904}{{\tt arXiv:1811.05904}}].

\bibitem{Kalyanapuram:2019nnf}
N.~Kalyanapuram, {\it {Stokes Polytopes and Intersection Theory}},  {\em Phys.
  Rev. D} {\bf 101} (2020), no.~10 105010,
  [\href{http://arxiv.org/abs/1910.12195}{{\tt arXiv:1910.12195}}].

\bibitem{Aneesh:2019cvt}
P.~B. Aneesh, P.~Banerjee, M.~Jagadale, R.~Rajan, A.~Laddha, and S.~Mahato,
  {\it {On positive geometries of quartic interactions: Stokes polytopes, lower
  forms on associahedra and world-sheet forms}},  {\em JHEP} {\bf 04} (2020)
  149, [\href{http://arxiv.org/abs/1911.06008}{{\tt arXiv:1911.06008}}].

\bibitem{Salvatori:2019phs}
G.~Salvatori and S.~Stanojevic, {\it {Scattering Amplitudes and Simple
  Canonical Forms for Simple Polytopes}},  {\em JHEP} {\bf 03} (2021) 067,
  [\href{http://arxiv.org/abs/1912.06125}{{\tt arXiv:1912.06125}}].

\bibitem{Srivastava:2020dly}
I.~Srivastava, {\it {Constraining the weights of Stokes polytopes using BCFW
  recursions for \ensuremath{\phi}$^{4}$}},  {\em JHEP} {\bf 04} (2021) 064,
  [\href{http://arxiv.org/abs/2005.12886}{{\tt arXiv:2005.12886}}].

\bibitem{Mizera:2017cqs}
S.~Mizera, {\it {Combinatorics and Topology of Kawai-Lewellen-Tye Relations}},
  {\em JHEP} {\bf 08} (2017) 097, [\href{http://arxiv.org/abs/1706.08527}{{\tt
  arXiv:1706.08527}}].

\bibitem{Arkani-Hamed:2017mur}
N.~Arkani-Hamed, Y.~Bai, S.~He, and G.~Yan, {\it {Scattering Forms and the
  Positive Geometry of Kinematics, Color and the Worldsheet}},  {\em JHEP} {\bf
  05} (2018) 096, [\href{http://arxiv.org/abs/1711.09102}{{\tt
  arXiv:1711.09102}}].

\bibitem{Cachazo:2021llu}
F.~Cachazo, {\it {Diagonally Embedded Sets of ${\rm Trop}^+G(2,n)$'s in ${\rm
  Trop}\, G(2,n)$: Is There a Critical Value of $n$?}},
  \href{http://arxiv.org/abs/2104.10628}{{\tt arXiv:2104.10628}}.

\bibitem{2017arXiv170309953M}
T.~{Manneville} and V.~{Pilaud}, {\it {Geometric realizations of the accordion
  complex of a dissection}},  {\em arXiv e-prints} (Mar., 2017)
  arXiv:1703.09953, [\href{http://arxiv.org/abs/1703.09953}{{\tt
  arXiv:1703.09953}}].

\bibitem{Raman:2019utu}
P.~Raman, {\it {The positive geometry for $\phi^{p}$ interactions}},  {\em
  JHEP} {\bf 10} (2019) 271, [\href{http://arxiv.org/abs/1906.02985}{{\tt
  arXiv:1906.02985}}].

\bibitem{Aneesh:2019ddi}
P.~B. Aneesh, M.~Jagadale, and N.~Kalyanapuram, {\it {Accordiohedra as positive
  geometries for generic scalar field theories}},  {\em Phys. Rev. D} {\bf 100}
  (2019), no.~10 106013, [\href{http://arxiv.org/abs/1906.12148}{{\tt
  arXiv:1906.12148}}].

\bibitem{Kojima:2020tox}
R.~Kojima, {\it {Weights and recursion relations for $\phi^p$ tree amplitudes
  from the positive geometry}},  {\em JHEP} {\bf 08} (2020) 054,
  [\href{http://arxiv.org/abs/2005.11006}{{\tt arXiv:2005.11006}}].

\bibitem{Kalyanapuram:2020vil}
N.~Kalyanapuram and R.~G. Jha, {\it {Positive Geometries for all Scalar
  Theories from Twisted Intersection Theory}},  {\em Phys. Rev. Res.} {\bf 2}
  (2020), no.~3 033119, [\href{http://arxiv.org/abs/2006.15359}{{\tt
  arXiv:2006.15359}}].

\bibitem{John:2020jww}
R.~R. John, R.~Kojima, and S.~Mahato, {\it {Weights, Recursion relations and
  Projective triangulations for Positive Geometry of scalar theories}},  {\em
  JHEP} {\bf 10} (2020) 037, [\href{http://arxiv.org/abs/2007.10974}{{\tt
  arXiv:2007.10974}}].

\bibitem{Kalyanapuram:2020axt}
N.~Kalyanapuram, {\it {On Polytopes and Generalizations of the KLT Relations}},
   {\em JHEP} {\bf 12} (2020) 057, [\href{http://arxiv.org/abs/2009.10114}{{\tt
  arXiv:2009.10114}}].

\bibitem{Jagadale:2021iab}
M.~Jagadale and A.~Laddha, {\it {Towards positive geometry of multi scalar
  field amplitudes. Accordiohedron and effective field theory}},  {\em JHEP}
  {\bf 04} (2022) 100, [\href{http://arxiv.org/abs/2104.04915}{{\tt
  arXiv:2104.04915}}].

\bibitem{Baadsgaard:2015ifa}
C.~Baadsgaard, N.~E.~J. Bjerrum-Bohr, J.~L. Bourjaily, and P.~H. Damgaard, {\it
  {Scattering Equations and Feynman Diagrams}},  {\em JHEP} {\bf 09} (2015)
  136, [\href{http://arxiv.org/abs/1507.00997}{{\tt arXiv:1507.00997}}].

\bibitem{Baadsgaard:2016fel}
C.~Baadsgaard, N.~E.~J. Bjerrum-Bohr, J.~L. Bourjaily, and P.~H. Damgaard, {\it
  {String-Like Dual Models for Scalar Theories}},  {\em JHEP} {\bf 12} (2016)
  019, [\href{http://arxiv.org/abs/1610.04228}{{\tt arXiv:1610.04228}}].

\bibitem{SSTrop}
D.~Speyer and B.~Sturmfels, {\it {The tropical Grassmannian}},  {\em Advances
  in Geometry} {\bf 4} (2004), no.~3 389--411,
  [\href{http://arxiv.org/abs/math/0304218}{{\tt math/0304218}}].

\bibitem{BHV}
L.~J. Billera, S.~P. Holmes, and K.~Vogtmann, {\it Geometry of the space of
  phylogenetic trees},  {\em Adv. Appl. Math.} {\bf 27} (Nov., 2001) 733--767.

\bibitem{Alex}
A.~{Postnikov}, {\it {Total positivity, Grassmannians, and networks}},  {\em
  arXiv Mathematics e-prints} (Sept., 2006) math/0609764,
  [\href{http://arxiv.org/abs/math/0609764}{{\tt math/0609764}}].

\bibitem{Cachazo:2013iea}
F.~Cachazo, S.~He, and E.~Y. Yuan, {\it {Scattering of Massless Particles:
  Scalars, Gluons and Gravitons}},  {\em JHEP} {\bf 07} (2014) 033,
  [\href{http://arxiv.org/abs/1309.0885}{{\tt arXiv:1309.0885}}].

\bibitem{Early:2019eun}
N.~Early, {\it {Planar kinematic invariants, matroid subdivisions and
  generalized Feynman diagrams}},  \href{http://arxiv.org/abs/1912.13513}{{\tt
  arXiv:1912.13513}}.

\bibitem{kChords}
D.~{Young}, {\it {Linear $k$-Chord Diagrams}},  {\em arXiv e-prints} (Apr.,
  2020) arXiv:2004.06921, [\href{http://arxiv.org/abs/2004.06921}{{\tt
  arXiv:2004.06921}}].

\bibitem{Brezin:1977sv}
E.~Brezin, C.~Itzykson, G.~Parisi, and J.~B. Zuber, {\it {Planar Diagrams}},
  {\em Commun. Math. Phys.} {\bf 59} (1978) 35.

\bibitem{Cachazo:2019ngv}
F.~Cachazo, N.~Early, A.~Guevara, and S.~Mizera, {\it {Scattering Equations:
  From Projective Spaces to Tropical Grassmannians}},
  \href{http://arxiv.org/abs/1903.08904}{{\tt arXiv:1903.08904}}.

\bibitem{Borges:2019csl}
F.~Borges and F.~Cachazo, {\it {Generalized Planar Feynman Diagrams:
  Collections}},  \href{http://arxiv.org/abs/1910.10674}{{\tt
  arXiv:1910.10674}}.

\bibitem{Cachazo:2019xjx}
F.~Cachazo, A.~Guevara, B.~Umbert, and Y.~Zhang, {\it {Planar Matrices and
  Arrays of Feynman Diagrams}},  \href{http://arxiv.org/abs/1912.09422}{{\tt
  arXiv:1912.09422}}.

\bibitem{Early:2020hap}
N.~Early, {\it {Weighted blade arrangements and the positive tropical
  Grassmannian}},  \href{http://arxiv.org/abs/2005.12305}{{\tt
  arXiv:2005.12305}}.

\bibitem{Early:2021tce}
N.~Early, {\it {Planarity in Generalized Scattering Amplitudes: PK Polytope,
  Generalized Root Systems and Worldsheet Associahedra}},
  \href{http://arxiv.org/abs/2106.07142}{{\tt arXiv:2106.07142}}.

\bibitem{oeis}
{OEIS Foundation Inc. (2019), The On-Line Encyclopedia of Integer Sequences. }
  {\em \url{http://oeis.org}}.

\bibitem{nara}
J.~C. {Novelli} and J.~Y. {Thibon}, {\it {Hopf Algebras of m-permutations,
  (m+1)-ary trees, and m-parking functions}},  {\em arXiv e-prints} (Mar.,
  2014) arXiv:1403.5962, [\href{http://arxiv.org/abs/1403.5962}{{\tt
  arXiv:1403.5962}}].

\end{thebibliography}\endgroup

\end{document}